\documentclass[twocolumn]{autart} 
\usepackage{cite}
\usepackage{amsmath,amssymb,amsfonts}
\usepackage{mathrsfs} 
\usepackage{algorithmic}
\usepackage{graphicx}
\usepackage{textcomp}
\usepackage{booktabs}
\usepackage{optidef} 
\usepackage{tikz}
\usepackage{eucal}
\usetikzlibrary{arrows.meta,positioning,quotes,shapes.multipart}
\usepackage{comment}
\usepackage{optidef}
\usepackage{MnSymbol}
\usepackage{schemabloc}
\usepackage[straightvoltages]{circuitikz}
\usepackage{tikz}
\usetikzlibrary{arrows,shapes,external,babel}
\usepackage{pgfplots}
\pgfplotsset{compat=1.18} 
\usepackage{placeins}
\usepackage{booktabs}
\usepackage[english]{babel}
\usepackage{xspace}
\usepackage[T1]{fontenc}

\newcommand{\Tr}{\ensuremath{\mathrm{Tr}}\xspace}

\newcommand{\hm}[1]{\ensuremath{\mathcal{#1}}\xspace}

\newcommand{\idn}[1][n]{\ensuremath{\mathrm{Id}_#1}\xspace}

\newcommand{\pddphi}{\ensuremath{\frac{\partial}{\partial \varphi}}\xspace}

\newcommand{\ji}{\ensuremath{\textsf{j}}\xspace}

\newcommand{\revision}[2]{%
 \ifstrequal{#1}{Maxime}{%
 {\color{purple}#2}%
 }{%
 \ifstrequal{#1}{Jamal}{%
  {\color{red}#2}%
 }{%
  \ifstrequal{#1}{Pierre}{%
  {\color{blue}#2}%
  }{%
  \ifstrequal{#1}{black}{%
   #2%
  }{%
   \ifstrequal{#1}{commented}{%
   }{%
   {\color{#1}#2}%
   }%
  }%
  }%
 }%
 }%
}

\DeclareMathOperator*{\diag}{diag}

\newtheorem{property}{Property}
\newtheorem{definition}{Definition}
\newtheorem{remark}{Remark}
\newtheorem{corollary}{Corollary}
\newtheorem{proposition}{Proposition}

\begin{document}
\begin{frontmatter}
\title{Harmonic Modeling and Control under Variable-Frequency}
\author[Cran,Safran]{Maxime Grosso}, 
\author[Cran]{Pierre Riedinger}, 
\author[Cran,IUF]{Jamal Daafouz}, 
\author[Lemta]{Serge~Pierfederici}, 
\author[Safran]{Hicham~Janati-Idrissi}, 
\author[Safran]{Blaise~Lap\^{o}tre}
\address[Cran]{Universit\'e de Lorraine, CNRS, CRAN, F-54000 Nancy, France} 
\address[Lemta]{Universit\'e de Lorraine, CNRS, LEMTA, F-54000 Nancy, France} 
\address[Safran]{Safran Electronics \& Defense, F-91344 Massy, France.} 
\address[IUF]{Institut Universitaire de France, F-75000 Paris, France.}
\begin{abstract}
This paper develops a harmonic-domain framework for systems with variable fundamental frequency. A variable-frequency sliding Fourier decomposition is introduced in the phase domain, together with necessary and sufficient conditions for time-domain realizability.  An exact harmonic-domain differential model is derived for general nonlinear systems under variable frequency, without assumptions on the frequency variation. An explicit parameter-varying approximation is then obtained, along with a tight error bound expressed in terms of local relative frequency variation, providing a non-conservative validity criterion and clarifying the limitations of classical heuristics. A main result shows that, for linear phase-periodic systems with affine frequency dependence, stability analysis and control synthesis can be carried out without approximation and without assumptions on the frequency variation, provided the frequency evolves within a prescribed interval. As a consequence, both problems reduce to harmonic Lyapunov inequalities evaluated at the two extreme frequency values, yielding a convex LMI characterization. The framework is illustrated on a variable-speed permanent magnet synchronous motor.
\end{abstract}
\end{frontmatter}
\section{Introduction}\label{sec:introduction}
\vspace{-.3cm}
Periodic signals can be naturally decomposed into harmonic components, each associated with an integer multiple of a fundamental frequency. This harmonic structure provides a useful lens for analyzing and modeling systems subject to periodic excitations or exhibiting periodic steady-state behavior. Over the past decades, two main classes of approaches have emerged to address such phenomena. The first class, rooted in control theory, focuses on the synthesis of control laws ensuring asymptotic tracking or rejection of periodic signals. 
It includes output regulation schemes based on the internal model principle, repetitive control, and their adaptive or robust extensions~\cite{Muramatsu2018PDOB,Han2023InternalModel,Astolfi2022HarmonicIM,Cecilia2024IMObserver,Chung2012AFC}. 
The second class, historically developed within electrical engineering, adopts a harmonic or frequency-domain viewpoint. 
Seminal contributions include harmonic state-space representations~\cite{sanders_generalized_1991,kwon_frequency-domain_2017} and dynamic phasors~\cite{Mattavelli97, stankovic_modeling_1999, levronModelingPowerNetworks2017, karamiDynamicPhasorbasedAnalysis2018} which have demonstrated their effectiveness in applications such as power electronics, electrical drives, and grid-connected converters~\cite{Cai2024HarmonicSSM,Kong2025DataLight}, where rich harmonic content challenges classical time-domain modeling. 
By representing periodic behaviors through the evolution of harmonic components, these approaches recast time-periodic dynamics into time-invariant (LTI) models in the harmonic domain, enabling the systematic use of standard LTI tools. 
 
Most existing harmonic modeling and control frameworks rely on the assumption of a constant fundamental frequency~\cite{Kwon17, ormrodHarmonicStateSpace2013, caiResearchHarmonicStateSpace2023, lyuHarmonicStateSpaceBased2019, Blin2022NecessaryTime}. This assumption is violated in many emerging applications where the fundamental frequency varies with operating conditions or is itself a control variable. Representative examples include variable-speed motor drives, aircraft electrical systems, renewable energy generation, and autonomous power systems. In such settings, frequency variations directly affect harmonic interactions, rendering fixed-frequency harmonic models inaccurate and potentially destabilizing when used for control design.
Early attempts to address variable-frequency effects introduced dynamic phasors with time-varying frequency~\cite{sanders_generalized_1991, yang_dynamic_2016}. These approaches typically relied on heuristic slowly-varying assumptions, introduced nonlinear or delayed terms, and lacked rigorous guarantees regarding time-domain realizability and control-oriented model validity. More recent contributions have clarified fundamental aspects of harmonic modeling, including the coincidence condition for time-domain realizability~\cite{Blin2022NecessaryTime} and computational tools for harmonic Lyapunov and LMI-based analysis~\cite{vernerey_tblmi_2025}. However, these advances remain largely restricted to the constant-frequency case, and a rigorous framework for harmonic modeling and control under time-varying frequency is still lacking.

To address this gap, we start from a sliding Fourier analysis formulated in the phase domain and introduce a variable-frequency harmonic decomposition, together with necessary and sufficient conditions for time-domain realizability. Building on this representation, we derive an exact harmonic-domain differential model that captures the true dynamics of harmonic components under frequency variation, without relying on slowly-varying frequency assumptions. This exact formulation serves as the foundation for a systematic parameter-varying approximation, together with a precise and computable error bound that quantifies its accuracy in terms of local relative frequency variation. Within the proposed approach, we further show that for a broad and practically relevant class of linear phase-periodic systems whose harmonic representation depends affinely on the instantaneous frequency, stability and state-feedback synthesis can be addressed exactly in the harmonic domain. In particular, we establish a Lyapunov stability condition that can be certified by evaluating harmonic Lyapunov inequalities at a finite number of frequency vertices, without imposing assumptions on the frequency variation. This result bridges harmonic modeling with LMI-based control design 
enabling stability guarantees over a prescribed frequency interval.
A preliminary version of this work appeared in~\cite{grossoFrequencyvaryingHarmonicDomain2025}. Here, we significantly extend that study by removing restrictive assumptions, deriving the harmonic formulation from first principles, and establishing precise conditions under which an exact and well-posed parameter-varying control framework can be obtained.

The remainder of the paper is organized as follows. Section~\ref{sec:preliminaries} recalls the foundations of harmonic modeling under the classical constant-frequency assumption. Section~\ref{sec:fourier_decomposition} introduces the proposed variable-frequency sliding Fourier decomposition and establishes its fundamental properties. Section~\ref{sec:exact_model} derives the exact harmonic-domain dynamical model governing the evolution of variable-frequency harmonic components. Section~\ref{sec:lpv_approximation} develops a parameter-varying approximation together with a quantitative and non-conservative validity criterion. Section~\ref{sec:robust_synthesis_LPP} establishes exact Lyapunov stability and control synthesis results for phase–periodic systems with affine frequency dependence. Finally, Section~\ref{sec:pmsm_application} illustrates the proposed framework on a variable-speed permanent magnet synchronous motor.

\vspace{-.3cm}
{\bf Notations: } The transpose of a matrix $A$ is denoted $A^\top$ and $A^*$ denotes the complex conjugate transpose $A^*=\bar A^\top$. The $n$-dimensional identity matrix is denoted $\idn$. The infinite identity matrix is denoted $\mathcal{I}$. $C^a$ denotes the space of absolutely continuous functions,
$L^{p}([a,b],\mathbb{C}^n)$ (resp. $\ell^p(\mathbb{C}^n)$) denotes the Lebesgue spaces of $p$-integrable functions on $[a, b]$ with values in $\mathbb{C}^n$ (resp. $p$-summable sequences of $\mathbb{C}^n$) for $1\leq p\leq\infty$. $L_{loc}^{p}$ is the set of locally $p$-integrable functions. Finally, $\otimes$ denotes the Kronecker product.
\section{Fixed-Frequency Harmonic Modeling}\label{sec:preliminaries}
Before introducing the proposed framework for variable-frequency systems, we briefly recall the fundamental principles of harmonic modeling under the classical assumption of a constant fundamental frequency~\cite{Blin2022NecessaryTime}. 
\begin{definition}\label{sfd_fixed}
For a time-signal $x \in L^2_{\text{loc}}(\mathbb{R}, \mathbb{C})$ and a fixed period $T_0 = 2\pi/\omega_0$, its sliding Fourier decomposition (SFD) is the mapping $\hm{F}_{T_0}: x \mapsto X\in C^a(\mathbb{R},\ell^2(\mathbb{C}^n))$ where for every $t$, the components of the sequence $X(t) = {(X_k(t))}_{k \in \mathbb{Z}}$, are given by:
\begin{equation}\label{eq:fourier_fixed_freq}
 X_k(t) = \frac{1}{T_0}\int_{t-T_0}^t x(\tau)e^{-\ji k \omega_0 \tau}d\tau.
\end{equation}
$X_k$ are called the $k$-th dynamic phasors.
For a time vector-valued function $x\in L^2_{loc}(\mathbb{R},\mathbb{C}^n)$ (or matrix-valued), this extends to
\begin{equation}
 X=\hm F(x)=(\hm F(x_1),\ldots,\hm F(x_n)).\label{fourier}
\end{equation}
\end{definition}
\vspace{-.3cm}
Conversely, the arrival space, denoted by $\mathcal H_{T_0}$ and commonly referred to as the harmonic domain~\cite{Blin2022NecessaryTime}, is characterized as follows.
\begin{thm}\label{thm:coincidence_fixed}
For a given $X \in C^a(\mathbb{R}, \ell^2(\mathbb{C}^n))$, there exists $x \in L^2_{\text{loc}}(\mathbb{R}, \mathbb{C}^n)$ such that $ X=\hm F_{T_0}(x)$ if and only if 
\begin{equation}\label{eq:coincidence_fixed_cond}
 \dot{X}_k(t) = \dot{X}_0(t) e^{-\ji k \omega_0 t}, \quad \text{a.e.\ for all } k \in \mathbb{Z}.
\end{equation}
This condition defines a subspace of $C^a(\mathbb{R}, \ell^2(\mathbb{C}^n))$ called the harmonic domain $\hm{H}_{T_0}$. The inverse mapping $\hm{F}_{T_0}^{-1}: \hm{H}_{T_0} \to L^2_{\text{loc}}$ is given for a continuous signal $x$ by:
\begin{equation}\label{recons}
 x(t) = \sum_{k \in \mathbb{Z}} X_k(t) e^{\ji k \omega_0 t} + \frac{{T_0}}{2}\dot{X}_0(t).
\end{equation}
\end{thm}
This theorem is commonly referred to as the \emph{coincidence condition}.
It characterizes the compatibility between the time evolution of the harmonic coefficients and an underlying time-domain signal, and guarantees that the harmonic representation admits a unique time-domain realization. We now recall the following result~\cite{Blin2022NecessaryTime}.
\begin{thm}\label{equiv_mod}
Consider a general dynamical system
\begin{equation}
 \dot{x}(t) = f(t, x(t))\label{st}
\end{equation} 
If $x$ is a Carathéodory solution of~\eqref{st} then, for any $T_0>0$, $X=\hm F_{T_0}(x)\in \hm{H}_{T_0}$ and is a solution of:
\begin{equation}\label{eq:harmodel_fixed_freq}
 \dot{X}(t) = \hm{F}_{T_0}(f(t, x(t))) - \omega_0\hm{N} X(t),
\end{equation}
where $\hm{N} = \idn \otimes \diag(\ji \mathbb{Z})$. 
Conversely, if $X$ is a solution of~\eqref{eq:harmodel_fixed_freq} that belongs to $\hm{H}_{T_0}$, its representative $x(t)$ is a solution of~\eqref{st}.
\end{thm}
This theorem is particularly useful for $T_0$-periodic systems (i.e., $f(t,x)=f(t+T_0,x)$), since in this case the harmonic representation~\eqref{eq:harmodel_fixed_freq} reduces to a time-invariant system. A key advantage is that $\hm F_{T_0}(f(t,x(t)))$ can then be computed explicitly when $f$ is polynomial in $x$ with periodic coefficients, using the following property~\cite{Blin2022NecessaryTime}.
\begin{property}\label{prod}
Let $x \in L^2_{\mathrm{loc}}$ and let $A$ and $B$ be $L^\infty$-matrix valued functions. Then, the following properties hold:
\begin{subequations}\label{eq:toeplitz_products}
\begin{align}
 \hm{F}_{T_0}(Ax) &= \hm{T}_{T_0}(A)\,\hm{F}_{T_0}(x) = \hm{A} X,\\
 \hm{T}_{T_0}(AB) &= \hm{T}_{T_0}(A)\,\hm{T}_{T_0}(B) = \hm{A}\,\hm{B},
\end{align}
\end{subequations}
where $\hm{T}_{T_0}$ is the Toeplitz transformation. For a scalar function $a \in L^\infty$ with Fourier coefficients ${(a_k)}_{k \in \mathbb{Z}}$, the corresponding infinite-dimensional Toeplitz operator, bounded on $\ell^2$, is defined as:
\begin{equation}\label{eq:toeplitz_matrix_def}
 \hm{T}_{T_0}(a)=
 \begin{bmatrix}
 \ddots & \vdots & \vdots & \vdots & \\
 \cdots & a_0 & a_{-1} & a_{-2} & \cdots \\
 \cdots & a_1 & a_0 & a_{-1} & \cdots \\
 \cdots & a_2 & a_1 & a_0 & \cdots \\
 & \vdots & \vdots & \vdots & \ddots
 \end{bmatrix}.
\end{equation}
For a matrix-valued function $A = {(a_{ij})}_{1 \le i \le n, 1 \le j \le m}$, the Toeplitz-block operator is
$ \hm{A} = {\big(\hm{T}_{T_0}(a_{ij})\big)}_{1 \le i \le n, 1 \le j \le m}$.
Moreover, the $\ell^2$-norm of $\hm{A}$ satisfies $
 \|\hm{A}\|_{\ell^2} = \|A\|_{L^\infty}$.
\end{property}
For illustration, consider an LTP system of the form
\begin{equation}\label{eq:basicLTP}
 \dot{x} = A(t)x + B(t)u
\end{equation} 
where $A(t)$ and $B(t)$ are ${T_0}$-periodic, $L^\infty$-matrix valued functions. The harmonic representation~\eqref{eq:harmodel_fixed_freq} reduces to an infinite-dimensional LTI system:
\begin{equation}\label{eq:ltp_to_lti}
 \dot{X}(t) = (\hm{A} - \omega_0 \hm{N})X(t) + \hm{B}U(t).
\end{equation}
This LTP-to-LTI transformation enables the application of classical LTI analysis and synthesis tools to periodic systems~\cite{grosso_control_2024,grosso_harmonic_2025}. The challenge is to extend this powerful and rigorous framework to systems in which the fundamental frequency, $\omega_0$, is no longer constant.
\section{The variable-frequency SFD}\label{sec:fourier_decomposition}
We now relax the constant-frequency assumption and consider a time-varying fundamental frequency, denoted $\omega(t)$. The central idea is to replace the fixed period $T_0$ with a time-varying ``pseudo-period'' $T(t)$, following the approach originally proposed in~\cite{sanders_generalized_1991,yang_dynamic_2016}.
In the following, we distinguish between time-domain variables $t$, phase-domain variables $\phi$, and the phase function $\theta(t)$ that maps time to phase.
\begin{definition}\label{def:phase_function}
Let $\omega \in C^0(\mathbb{R}, \mathbb{R}^+)$ be a strictly positive, continuous function. The associated \emph{phase function} $\theta$ is a $C^1$ diffeomorphism defined by
\begin{equation}\label{theta}
 \theta(t) \triangleq \theta(0) + \int_0^t \omega(s)\, ds.
\end{equation}
Its inverse, called the \emph{phase-to-time mapping}, is denoted by $p(\phi) = \theta^{-1}(\phi)$ and satisfies
$ \frac{d}{d\phi} p(\phi) = \frac{1}{\omega(p(\phi))}. $
\end{definition}
\vspace{-.3cm}
This diffeomorphism provides a natural way to express any time-domain signal in terms of its phase.
\begin{definition}\label{def:phase_domain_signal}
Let $x \in L^2_{\mathrm{loc}}(\mathbb{R}, \mathbb{C}^n)$. Its \emph{phase-domain representation} is the mapping
\[
\phi \mapsto x_\theta(\phi) \triangleq x(p(\phi)),
\]
so that for $\phi = \theta(t)$, we have: $x(t) = x_\theta(\theta(t))$.
\end{definition}
Importantly, this transformation is not a mere time reparameterization: it enables the definition of a fixed-length sliding Fourier window in the phase domain, which is essential for extending harmonic analysis to variable-frequency signals. The definition of the sliding Fourier decomposition for variable-frequency signals is recalled below~\cite[Appendix A]{sanders_generalized_1991}.
\begin{definition}
Let $\theta$ be a phase function as defined in~\eqref{theta}. The SFD of a time-domain signal $x$ is denoted by $\hm{F}_\theta$ and defined by $t \mapsto X(t):= X_\theta(\theta(t)), $
where $
X_\theta(\varphi) = \hm{F}_{2\pi}(x_\theta)(\varphi)$ and $\hm{F}_{2\pi}$ is the standard SFD, introduced in Definition~\ref{sfd_fixed}, over a window of length $2\pi$. The resulting $X_\theta(\varphi)$ is called the \emph{phase-domain phasor}.
\end{definition}
From this definition, we have the following property.
\begin{property}\label{prop:defkPhasorfvar}
The $k$-th phasor of $x$ is given by
\begin{align}
 X_k(t) &= \frac{1}{2\pi} \int_{\theta(t)-2\pi}^{\theta(t)} x_\theta(\phi) \, e^{-\ji k\phi} \, d\phi \label{eq:dp_def_phase} \\
  &= \frac{1}{2\pi} \int_{t-T(t)}^{t} x(\tau) \, e^{-\ji k\theta(\tau)} \, \omega(\tau) \, d\tau, \label{eq:dp_def_time}
\end{align}
where the time-varying \emph{pseudo-period} $T(t)$ is defined as the unique positive solution to
\begin{equation}\label{eq:pseudo_period_def}
 \theta(t) - \theta(t - T(t)) = 2\pi, \quad \forall t \in \mathbb{R}.
\end{equation}
Moreover, if $\omega$ is continuous ($\omega \in C^0$), then $T(t)$ is a $C^1$ function satisfying
\begin{equation}\label{eq:pseudo_period_derivative}
 \dot{T}(t) = 1 - \frac{\omega(t)}{\omega(t - T(t))}.
\end{equation}
\end{property}
\vspace{-.5cm}
\begin{pf}
Applying Definition~\ref{sfd_fixed} to $\hm{F}_{2\pi}(x_\theta)(\varphi)$ yields
\[
X_k(\varphi) = \frac{1}{2\pi} \int_{\varphi-2\pi}^{\varphi} x_\theta(\phi) \, e^{-\ji k \phi} \, d\phi,
\]
noting that in~\eqref{eq:fourier_fixed_freq} we have $\omega_0 = 1$ when $T_0 = 2\pi$. 
Equation~\eqref{eq:dp_def_phase} follows directly by substituting $\varphi = \theta(t)$. 
Equations~\eqref{eq:dp_def_time} and~\eqref{eq:pseudo_period_def} are obtained by performing the change of variables $\phi = \theta(\tau)$. 
Finally,  differentiating~\eqref{eq:pseudo_period_def} with respect to $t$ concludes the proof. \hfill $\square$
\end{pf}
\begin{remark}
When $\omega$ is constant,~\eqref{eq:dp_def_time} reduces to the classical fixed-frequency Fourier decomposition~\eqref{eq:fourier_fixed_freq}. If $\omega(t)$ is \emph{phase-periodic}, the right-hand side of~\eqref{eq:pseudo_period_derivative} vanishes, implying that the pseudo-period $T(t)$ is constant; denote this constant by $T_c$. Differentiating~\eqref{eq:pseudo_period_def} then shows that $\omega(t)$ is also $T_c$-periodic. Since $T_c$ generally differs from $2\pi/\omega(t)$, the constant phasor associated with $\omega(t)$ through~\eqref{eq:dp_def_time} should not be confused with the constant phasor obtained from the classical Fourier decomposition~\eqref{eq:fourier_fixed_freq} with $T_0=T_c$. These quantities are defined with respect to different referentials, namely the phase and the time domains, and therefore represent fundamentally distinct objects.
\end{remark}
The following theorem generalizes the classical coincidence condition of Theorem~\ref{thm:coincidence_fixed} to the variable-frequency setting and establishes a necessary and sufficient condition for time-domain realizability.
\begin{thm}\label{thm:coincidence}
Let $X \in C^a(\mathbb{R}, \ell^2(\mathbb{C}^n))$. There exists a signal $x \in L^2_{\mathrm{loc}}(\mathbb{R}, \mathbb{C}^n)$ such that $X = \hm{F}_\theta(x)$ if and only if
\begin{equation}\label{eq:coincidence_varying}
 \dot{X}_k(t) = \dot{X}_0(t) \, e^{-\ji k \theta(t)}, \quad \text{for a.e. } t \text{ and all } k \in \mathbb{Z}.
\end{equation}
The subspace of $C^a(\mathbb{R}, \ell^2(\mathbb{C}^n))$ satisfying this condition is called the \emph{variable-frequency harmonic domain} and denoted $\hm{H}_\theta$. The inverse mapping $\hm{F}_\theta^{-1}: \hm{H}_\theta \to L^2_{\mathrm{loc}}$ is given, for a continuous signal $x$, by
\begin{equation}\label{recons_timedomain}
 x(t) = \sum_{k \in \mathbb{Z}} X_k(t) \, e^{\ji k \theta(t)} + \frac{\pi}{\omega(t)} \, \dot{X}_0(t).
\end{equation}
\end{thm}
\vspace{-.5cm}
\begin{pf}
By definition, for every $t$ we have $X(t):= X_\theta(\theta(t))$. Setting $\varphi = \theta(t)$ gives $t = p(\varphi)$, so that
$X(p(\varphi)) = X_\theta(\varphi), $
which expresses the phasor sequence as a function of the phase $\varphi$. Since $X_\theta(\varphi) = \hm{F}_{2\pi}(x_\theta)(\varphi)$, its phasor-derivative with respect to $\varphi$ satisfies, for every $k \in \mathbb{Z}$ (see Theorem~\ref{thm:coincidence_fixed}):
\begin{equation}\label{eq:coincidence_phase_domain}
 \frac{d}{d\varphi} X_{\theta,k}(\varphi) = \frac{d}{d\varphi} X_{\theta,0}(\varphi) \, e^{-\ji k \varphi}.
\end{equation}
On the other hand, the chain rule gives, for every $k \in \mathbb{Z}$,
\[
\dot{X}_k(t) = \omega(t) \frac{d}{d\varphi} X_{\theta,k}(\varphi) \Big|_{\varphi = \theta(t)}.
\]
Combining this with~\eqref{eq:coincidence_phase_domain} immediately yields Eq.~\eqref{eq:coincidence_varying}. 
Similarly, for a continuous phase-domain signal $x_\theta$, we have (see~\eqref{recons} with $T_0 = 2\pi$ and $\omega_0 = 1$):
\[
x_\theta(\varphi) = \sum_{k \in \mathbb{Z}} X_{\theta,k}(\varphi) \, e^{\ji k \varphi} + \pi \frac{d}{d\varphi} X_{\theta,0}(\varphi),
\]
which directly yields the formula~\eqref{recons_timedomain}. \hfill $\square$
\end{pf}
\section{Exact Variable-Frequency Harmonic Model}\label{sec:exact_model}
Building on the variable-frequency harmonic representation introduced in the previous section, we now derive an exact and compact harmonic-domain differential model describing the evolution of variable-frequency phasors for general nonlinear, time-varying systems. This result provides a rigorous foundation for both approximation and control synthesis developed in the sequel.
Consider a general nonlinear time-varying system
\begin{equation}\label{eq:time_domain_system}
 \dot{x}(t) = f(t, x(t)),
\end{equation}
where $x(t) \in \mathbb{R}^n$ and $f$ satisfies the Carathéodory conditions ensuring existence of solutions. Using the phase-domain representation $x_\theta(\phi) = x(p(\phi))$, system~\eqref{eq:time_domain_system} can be equivalently expressed in the phase domain as
\begin{equation}\label{eq:phase_domain_system}
 \frac{d}{d\phi} x_\theta(\phi)
 = \frac{f(p(\phi), x_\theta(\phi))}{\omega(p(\phi))}.
\end{equation}
Applying Theorem~\ref{equiv_mod} to the phase-domain system yields the infinite-dimensional phasor dynamics
\begin{equation}\label{eq:phase_phasor_dyn_eq}
 \frac{d}{d\phi} X_\theta(\phi)
 = \hm{F}_{2\pi}\!\left(\frac{f_\theta(\cdot, x_\theta(\cdot))}{\omega_\theta(\cdot)}\right)(\phi)
 - \hm{N} X_\theta(\phi),
\end{equation}
where $f_\theta(\phi, x_\theta(\phi)) = f(p(\phi), x_\theta(\phi))$ and
$\omega_\theta(\phi) = \omega(p(\phi))$.
We now translate this phase-domain formulation back to the time domain, yielding an exact variable-frequency harmonic model.
\begin{thm}\label{thm:exact_model}
The exact dynamics of the variable-frequency phasor 
$X(t) = \hm{F}_\theta(x)(t)$ associated with~\eqref{eq:time_domain_system} are given by
\begin{equation}\label{eq:exact_model_final}
 \dot{X}(t)
 = \hm{G}(\omega(t))\,\hm{F}_\theta\!\bigl(f(t,x(t))\bigr)
 - \omega(t)\,\hm{N} X(t),
\end{equation}
where $ \hm{G}(\omega(t)) = \idn \otimes \omega(t)\,\hm{T}_\theta^{-1}(\omega)(t)$, and $ \hm{N} = \idn \otimes \diag_{k \in \mathbb{Z}}(\ji k)$.
\end{thm}
\vspace{-.5cm}
\begin{pf}
By definition, $X(t)=X_\theta(\theta(t))$. Applying the chain rule yields
$
\dot{X}(t)
= \omega(t)\frac{d}{d\phi}X_\theta(\phi)\Big|_{\phi=\theta(t)}.
$
Substituting~\eqref{eq:phase_phasor_dyn_eq} gives
\[
\dot{X}(t)
= \omega(t)\hm{F}_{2\pi}\!\left(\frac{f_\theta}{\omega_\theta}\right)(\theta(t))
- \omega(t)\hm{N}X(t).
\]
Since
$
\hm{F}_{2\pi}\!\left(\frac{f_\theta}{\omega_\theta}\right)(\theta(t))
= \hm{F}_\theta\!\left(\frac{f(t,x(t))}{\omega(t)}\right),
$
Property~\ref{prod} yields
\[
\hm{F}_\theta\left(\frac{f}{\omega}\right) 
 = \hm{T}_\theta\big(\idn / \omega \big) \, \hm{F}_\theta(f) 
 = (\idn \otimes \hm{T}_\theta^{-1}(\omega)) \, \hm{F}_\theta(f).
\]
The result follows by linearity and noting that $\omega(t) \hm{N} X_\theta(\theta(t)) = \omega(t) \hm{N} X(t)$. \hfill $\square$
\end{pf}
When $\omega(t)\equiv\omega_0$, $\hm{T}_\theta(\omega)=\omega_0\hm I$ and~\eqref{eq:exact_model_final} reduces exactly to the classical fixed-frequency harmonic model.
We conclude this section by relating~\eqref{eq:exact_model_final} to existing exact formulations of variable-frequency dynamic phasors.
{In~\cite[Appendix A]{sanders_generalized_1991}, the evolution of the phasors are given by
\begin{equation}\label{eq:SandersV}
\begin{aligned}
 \dot X(t) =&\;
 x(t-T(t)\omega(t-T(t))\dot T(t))e^{-\hm N \theta(t)} \\
 &+ \hm F_\theta\big(f(t,x(t) )\big) + \left(\hm T_\theta(\tfrac{\dot \omega}{ \omega})- \hm N \hm T_\theta(\omega) \right)X.
\end{aligned}
\end{equation}}
Although~\eqref{eq:SandersV} is equivalent to~\eqref{eq:exact_model_final}
under suitable regularity assumptions, it involves explicit delays,
nonlinear couplings, and dependence on $\dot{\omega}(t)$.
In contrast,~\eqref{eq:exact_model_final} concentrates all variable-frequency
effects into the single operator $\hm{G}(\omega(t))$, avoids explicit delays and
frequency derivatives, and requires only $C^1$ regularity of the phase
$\theta(t)$, making it particularly amenable to control design.
\section{Parameter-Varying (PV) Approximation}\label{sec:lpv_approximation}
While the exact model of Theorem~\ref{thm:exact_model} provides a complete
description of variable-frequency harmonic dynamics, its direct use in control
synthesis is generally impractical due to the time-varying operator
$\hm{G}(\omega(t))$, which depends non-linearly on the frequency signal
$\omega$. This section derives a tractable parameter-varying approximation and
establishes a precise and computable bound on the resulting modeling error.
\subsection{An approximated PV harmonic model}
Starting from the exact model~\eqref{eq:exact_model_final}, the dynamics can be
decomposed as
\begin{equation}\label{eq:pv_decomposition}
 \dot{X}(t)
 = \underbrace{\bigl(\hm{F}_\theta(f(t,x)) - \omega(t)\hm{N}X(t)\bigr)}_{\text{Nominal part}}
 + \hm{E}(t),
\end{equation}
where the error term $ \hm{E}(t)$ is given by
\begin{equation}\label{eq:pv_error}
 \hm{E}(t)
 = (\idn\otimes\Delta_\omega(t))\,\hm{F}_\theta(f(t,x)),
\end{equation} 
with $ \Delta_\omega(t) = \omega(t)\hm{T}_\theta^{-1}(\omega)(t) - \hm{I}$.
The nominal part of~\eqref{eq:pv_decomposition} has the structure of a parameter-varying system whereas $\hm{E}(t)$ captures the deviation with respect to the exact model~\eqref{eq:exact_model_final}. In order to quantify this deviation, we introduce the following notion.
\begin{definition} 
For a given time $t$ and $\tau \in [t - T(t), \, t]$, the relative variation of the frequency $\omega$ at time $\tau$ with respect to its value at time $t$ is defined as
\begin{equation}\label{eq:relative_variation}
 \delta_\omega(\tau, t) \triangleq \frac{\omega(t) - \omega(\tau)}{\omega(\tau)}.
\end{equation}
\end{definition}
The following theorem establishes a direct link between the Toeplitz operator $\Delta_\omega(t)$ and the relative frequency variation over the sliding pseudo-period, and provides a simple characterization of its induced $\ell^2$-norm.
\begin{thm}\label{thm:lpv_approximation}
For every $t$,
\begin{equation}\label{eq:delta_omega}
 \Delta_\omega(t)
 = \hm{T}_\theta(\delta_\omega(\cdot,t))(t),
\end{equation}
and its induced $\ell^2$-operator norm satisfies
\begin{equation}\label{eq:epsilon_criterion}
 \epsilon(t)
:= \|\Delta_\omega(t)\|_{\ell^2}
 = \sup_{\tau\in[t-T(t),t]}
 \left|\frac{\omega(t)-\omega(\tau)}{\omega(\tau)}\right|.
\end{equation}
\end{thm}
\vspace{-.5cm}
\begin{pf}
Using the linearity of the Toeplitz operator and as $\omega(t)$ is constant with respect to the integration variable $\tau$ in the Fourier decomposition, we have:
\begin{align*}
 \hm{T}_\theta(\delta_\omega(\cdot, t))(t) 
 &= \hm{T}_\theta\left(\frac{\omega(t)}{\omega(\cdot)} - 1\right)(t) \\
 &= \omega(t) \, \hm{T}_\theta^{-1}(\omega)(t) - \hm{I} = \Delta_\omega(t).
\end{align*}
Finally,~\eqref{eq:epsilon_criterion} follows from the relation (see Property~\ref{prod} or~\cite{bottcher_introduction_1999}):
$ \Delta_\omega(t)\|_{\ell^2} = \|\delta_\omega(\cdot, t)\|_{L^\infty}$.
\hfill $\square$
\end{pf}
The quantity $\epsilon(t)$, given by~\eqref{eq:epsilon_criterion}, measures the maximal relative frequency variation
over one pseudo-period. When $\epsilon(t)$ is small, the system locally behaves
as if the frequency were constant. Motivated by this observation, the following simplified model is obtained from the exact model~\eqref{eq:exact_model_final} by applying the approximation
\begin{equation}\label{eq:main_approximation}
\omega(t)\,\hm{T}_\theta^{-1}(\omega)(t) \approx \hm{I}.
\end{equation}
\begin{definition}\label{APVHM}
Under the condition $\epsilon(t)<\varepsilon$ for some positive $\varepsilon$, the
approximated PV harmonic model is given by
\begin{equation}\label{eq:lpv_model}
 \dot{X}(t)
 \approx \hm{F}_\theta(f(t,x)) - \omega(t)\hm{N}X(t),
\end{equation}
with $\omega(t)$ treated as a time-varying parameter.
\end{definition}
\begin{thm}\label{relative_error}
The relative modeling error induced by the PV approximation satisfies
\begin{equation}\label{error}
 \frac{\| \hm{E}(t)\|_{\ell^2}}
 {\|\hm{F}_\theta(f(t,x))\|_{\ell^2}}
 \le \epsilon(t),
\end{equation}
with $\epsilon(t)=0$ in the constant-frequency case.
\end{thm}
\vspace{-.5cm}
\begin{pf}
The result follows directly from~\eqref{eq:pv_error} and
$\|\idn\otimes\Delta_\omega(t)\|_{\ell^2}
=\|\Delta_\omega(t)\|_{\ell^2}$. \hfill $\square$
\end{pf}
\subsection{Analysis and comparison with classical heuristics}
We now analyze the behavior of the error measure $\epsilon(t)$ introduced in
Section~\ref{sec:lpv_approximation} and clarify its relationship with the
heuristic assumptions commonly used to justify PV approximations in the
literature. Our objective is to show that $\epsilon(t)$ provides a precise and
non-conservative quantification of approximation accuracy. Existing justifications of PV approximations {of model~\eqref{eq:SandersV}} typically rely on an informal {``slowly varying $\omega$''} assumption~\cite{sanders_generalized_1991, yang_dynamic_2016}, leading to:
\begin{enumerate}
 \item $\dot{T}(t) \approx 0$,
 \item $\omega(t) \approx \omega(t-T(t))$,
 \item $\dot{\omega}(t)/\omega(t) \approx 0$.
\end{enumerate}
The precise implications of these assumptions for the modeling error are rarely made explicit. From~\eqref{eq:pseudo_period_derivative}, the condition $\dot{T}(t)\approx 0$
implies $\omega(t)\approx \omega(t-T(t))$, which in turn corresponds to $ \delta_\omega(t-T(t),t) \approx 0$.
This reasoning, however, only constrains the frequency mismatch at the
\emph{boundary} of the pseudo-period. It does not control the frequency
variations occurring inside the interval $[t-T(t),\,t]$. In contrast, 
the proposed quantity $\epsilon(t)$ uniformly bounds the relative frequency variation over the \emph{entire}
pseudo-period. As such, $\epsilon(t)\approx 0$ constitutes a 
non-conservative condition for the validity of the PV approximation.

To make this distinction explicit, consider a linearly increasing frequency
profile $\omega(t)=at+\omega_0$, with $a=\dot{\omega}>0$. In this case, the error admits the closed-form expression
\begin{equation}\label{eq:epsilon_ramp}
\epsilon(t)
={\left(1-\frac{4\pi a}{{\omega(t)}^2}\right)}^{-\frac{1}{2}}-1
\approx \frac{2\pi \dot{\omega(t)}}{{\omega(t)}^2}.
\end{equation}
This expression shows that $\epsilon(t)$ decreases monotonically with $\omega(t)$.
Hence, for a constant acceleration, the PV approximation becomes increasingly
accurate at higher frequencies.

For a prescribed admissible modeling error $\varepsilon$, solving
$\epsilon(t)=\varepsilon$ in~\eqref{eq:epsilon_ramp} yields the maximum allowable
frequency rate
\begin{equation}\label{amax}
a_{\max}
=\frac{{\omega(t)}^2}{4\pi}\left(1-\frac{1}{{(1+\varepsilon)}^{2}}\right).
\end{equation}
This relation makes explicit how the admissible rate of frequency variation
scales with $\omega(t)$, a dependence that is not captured by classical
heuristics. Enforcing $\epsilon(t)=\varepsilon$ for all $t$ leads to the differential
equation
\begin{equation}\label{eq:hyperbolic_omega_de}
\dot{\omega}(t)
=K{\omega(t)}^2,
\quad
K=\frac{1}{4\pi}\left(1-\frac{1}{{(1+\varepsilon)}^{2}}\right),
\end{equation}
whose solution
$\omega(t)={({\omega_0}^{-1}-Kt)}^{-1}$ exhibits finite-time blow-up. In this case,
$ \frac{\dot{\omega}(t)}{\omega(t)} = K\omega(t)$,
which grows unboundedly with $\omega(t)$, even though $\epsilon(t)$ is held
constant by construction. This example demonstrates that the heuristic condition $\dot{\omega}/\omega \approx 0$ does not reliably characterize model validity. Figure~\ref{fig:compare_criteria} provides simulations with more complex frequency profiles to confirm this observation: both $|\dot{\omega}|$ and $|\dot{\omega}/\omega|$ may take large values while
$\epsilon(t)$ decreases, indicating improved approximation accuracy.

Overall, these results demonstrate that $\epsilon(t)$ provides a rigorous and
non-conservative characterization of the \emph{small local relative frequency
variation} assumption underlying the PV harmonic framework. Unlike classical
heuristics, it directly quantifies the modeling error induced by frequency
variations over the pseudo-period.
\begin{figure}[t]
 \centering
 \includegraphics[width=\columnwidth]{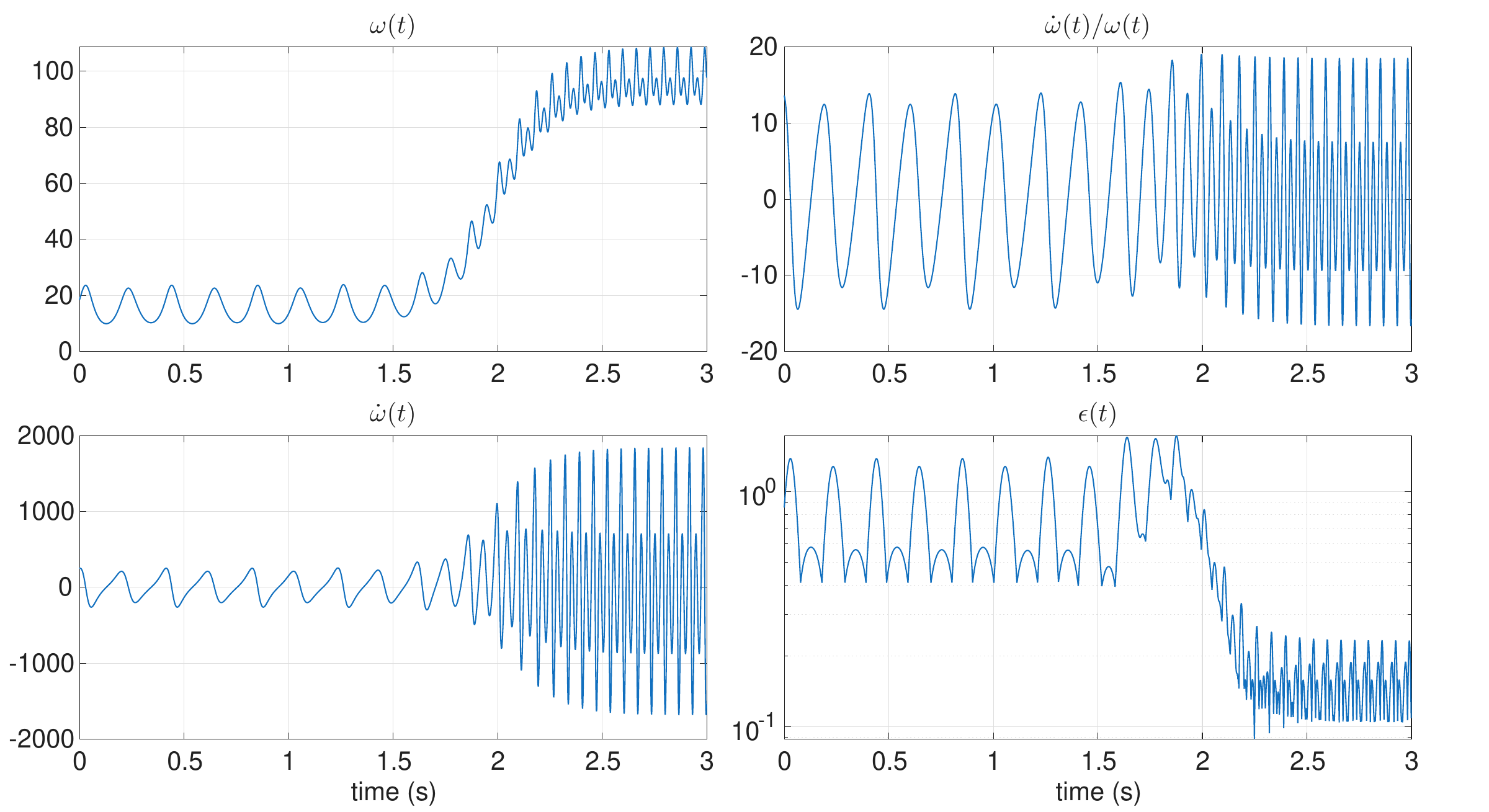}
 \caption{$\epsilon(t)$ provides an accurate measure of the modeling error, whereas the heuristic $\vert \dot{\omega}/\omega\vert $ and $\dot\omega$ are poor indicators.}\label{fig:compare_criteria}
\end{figure}
\section{Affine Linear phase-periodic Systems}\label{sec:robust_synthesis_LPP}
This section shows that, for a practically relevant class of phase-periodic systems, variable-frequency effects can be handled exactly in the harmonic domain. In contrast to the approximation-based approach of Section~\ref{sec:lpv_approximation}, we focus on linear phase-periodic systems whose harmonic representation depends affinely on the instantaneous frequency. Such a structure naturally arises in applications such as electrical drives, aircraft power systems, and modulated converters. For this class, we establish an exact Lyapunov stability condition for the autonomous dynamics, which then serves as the foundation for state-feedback controller synthesis via an appropriate system augmentation.
\subsection{Exact Lyapunov stability condition}
We consider linear systems whose coefficients are periodic in the phase and depend affinely on the instantaneous frequency. These systems are described by the following affinely frequency-modulated phase-periodic (AFM-LPP) model.
\begin{definition} AFM-LPP systems are described by
\begin{equation}\label{eq:afmlpp_time_domain}
\begin{aligned}
 \dot{x}(t) =& \big(A_0(\theta(t)) + \omega(t) A_1(\theta(t))\big)x(t) \\ 
 &+ \big(B_0(\theta(t)) + \omega(t) B_1(\theta(t))\big)u(t),
\end{aligned}
\end{equation}
where $A_i(\cdot)$ and $B_i(\cdot)$ are $L^\infty$ matrix-valued functions that are $2\pi$-periodic in $\theta$, and $\theta$ denotes a phase function as defined in Definition~\ref{def:phase_function}.
\end{definition}
When $A_1\equiv0$ and $B_1\equiv0$, the model reduces to a classical linear phase-periodic system. The AFM-LPP structure therefore strictly generalizes the fixed-frequency case while preserving a form that admits exact harmonic-domain analysis. 

Applying the variable-frequency SFD to~\eqref{eq:afmlpp_time_domain} and invoking Theorem~\ref{thm:exact_model} yields the exact harmonic representation
\begin{equation}\label{eq:HarmExact_full}
\begin{aligned}
 \dot{X}(t) =& \hm{G}(\omega(t))\big(\hm{A}_0 X(t) + \hm{B}_0 U(t)\big) \\
 &+ \omega(t)\big(\hm{A}_1 X(t) + \hm{B}_1 U(t) - \hm{N} X(t)\big),
\end{aligned}
\end{equation}
where the Toeplitz block operators $\hm{A}_i=\hm{T}_\theta(A_i)$ and
$\hm{B}_i=\hm{T}_\theta(B_i)$, $i=0,1$, are constant and bounded on $\ell^2$.

We first focus on the autonomous dynamics
\begin{equation}\label{eq:auto}
\dot{x}(t)=\big(A_0(\theta(t))+\omega(t)A_1(\theta(t))\big)x(t),
\end{equation}
with associated harmonic model
\begin{equation}\label{eq:harmo}
\dot{X}(t)=\big(\hm{G}(\omega(t))\hm{A}_0+\omega(t)(\hm{A}_1-\hm{N})\big)X(t).
\end{equation}
The following theorem provides an exact Lyapunov stability condition for these dynamics. Its key feature is that stability under arbitrary admissible frequency trajectories can be certified by enforcing Lyapunov inequalities at a finite number of frequency vertices only. This is made possible by two structural properties induced by affine frequency dependence: convexity with respect to $\omega$, and commutativity of Toeplitz operators generated by scalar phase functions.
\begin{thm}\label{main}
Consider the constant Toeplitz block operators $\hm A_0$ and $\hm A_1$, bounded on $\ell^2$.
Assume that there exists a Toeplitz block operator $\hm P=\hm P^*\succ 0$, bounded on $\ell^2$, such that
\begin{align}\label{eq:PolyStab}
\hm P\big(\hm A_0 + \omega_i(\hm A_1 - \hm N)\big)
+ {\big(\hm A_0 + \omega_i(\hm A_1 - \hm N)\big)}^*\hm P \prec 0,
\end{align}
for all $\omega_i \in \{\omega_{\min},\omega_{\max}\}$.
Let $\omega(t)$ be any time-varying frequency of the form
\[
\omega(t)=\alpha(t)\,\omega_{\min}+\big(1-\alpha(t)\big)\,\omega_{\max},
\qquad \alpha(t)\in(0,1).
\]
Then the harmonic Lyapunov inequality
\begin{align}\label{e2}
\hm P\big(\hm{G}(\omega(t))\hm A_0 + \omega(t)(\hm A_1 - \hm N)\big)+{(\cdot)}^* \prec 0
\end{align}
holds for all $t$.
Moreover, since $\hm P$, $\hm A_0$, and $\hm A_1$ admit $2\pi$-periodic and $L^\infty$ representatives,~\eqref{e2} is equivalent to the following differential Lyapunov inequality, which holds almost everywhere:
\begin{equation}\label{angleLyap}
\dot P(\theta)
+ P(\theta)\big(A_0(\theta) + \omega(t)A_1(\theta)\big)
+ {(\cdot)}^\top \prec 0,
\end{equation}
with $P(\theta)={P(\theta)}^\top\succ0$ and $\dot \theta(t)=\omega(t)$.
\end{thm}
\vspace{-.5cm}
\begin{pf}
Consider the Toeplitz transformation:
\[
\hm T_\theta(\omega)
= \hm T_\theta(\alpha)\,\omega_{\min}
+ \big(\hm I-\hm T_\theta(\alpha)\big)\,\omega_{\max}.
\]
Recall that the product of two positive definite operators remains positive definite if these operators commute.
Since $\alpha(t)\in(0,1)$, both operators $\hm T_\theta(\alpha)$ and
$\hm I-\hm T_\theta(\alpha)$ are positive definite.
Moreover, as $\alpha$ is a scalar function, $\hm T_\theta(\alpha)$ commutes with any Toeplitz operator and $\hm T_\theta(\alpha)={\hm T_\theta(\alpha)}^*$.
Using these properties, we multiply~\eqref{eq:PolyStab} on the left by
$\idn\otimes\hm T_\theta(\alpha)$ for $\omega_i=\omega_{\min}$, and by
$\idn\otimes(\hm I-\hm T_\theta(\alpha))$ for $\omega_i=\omega_{\max}$.
Summing the resulting inequalities preserves negative definiteness and yields
\begin{equation}\label{eq:inter}
\hm P\big(\hm A_0 + (\idn\otimes\hm T_\theta(\omega))(\hm A_1 - \hm N)\big)
+{(\cdot)}^* \prec 0.
\end{equation}
Multiplying~\eqref{eq:inter} on the left by
$$\hm{G}(\omega(t))=\omega(t){(\idn\otimes\hm T_\theta(\omega))}^{-1},$$
which is positive definite, self-adjoint and commutes with the involved operators,
directly gives~\eqref{e2}. Finally, since $\hm A_0$, $\hm A_1$, and $\hm P$ admit $2\pi$-periodic representatives,
multiplying~\eqref{eq:inter} on the left by ${(\idn\otimes\hm T_\theta(\omega))}^{-1}$ leads to
\[
\hm P\left[
\hm T_{2\pi}\!\left(\frac{A_0(\phi)}{\omega(p(\phi))}+A_1(\phi)\right)
-\hm N
\right]
+{(\cdot)}^* \prec 0.
\]
This Lyapunov inequality admits an equivalent differential formulation in the phase domain~\cite{vernerey_tblmi_2025}:
\[
P(\phi)\!\left(\frac{A_0(\phi)}{\omega(p(\phi))}+A_1(\phi)\right)
+{ (\cdot)}^\top
+\frac{dP(\phi)}{d\phi} \prec 0.
\]
The inequality~\eqref{angleLyap} follows by multiplication by $\omega(t)$
and application of the chain rule $\dot P = \omega\,\frac{dP}{d\phi}$. \hfill $\square$
\end{pf}
\begin{corollary}\label{coroStabLyap}
Let $\hm P$ and $P(\theta)$ satisfy the conditions of Theorem~\ref{main}.
Then the Lyapunov functions $\hm V(X(t)) = {X(t)}^* \hm P X(t) $ and $ v(t,x) = {x(t)}^\top P(\theta(t)) x(t)$
guarantee asymptotic stability of the systems~\eqref{eq:harmo} and~\eqref{eq:auto}, respectively,
for any frequency function $\omega(t)$ taking values in $[\omega_{\min},\omega_{\max}]$.
\end{corollary}
\subsection{Harmonic control design}\label{control_design}
We now extend the stability result of Theorem~\ref{main} to control design. To this end, we introduce an augmented system representation that incorporates the controller dynamics while preserving the AFM-LPP structure. As a result, the vertex-based Lyapunov conditions of Theorem~\ref{main} remain applicable to the closed-loop system.
The guiding idea is to incorporate integral action and harmonic cancellation systematically while remaining fully compatible with the Toeplitz operator structure. To this end, we introduce a forwarding variable $z$, whose
dynamics implement an internal-model principle:
\[
\dot{z}(t)=\omega(t)\big(J(\theta(t))z(t)+L(\theta(t))C(\theta(t))x(t)\big),
\]
where $J,L,C$ are given $2\pi$–periodic phase functions. In general, the matrix $C$ defines the regulated outputs, while
$J$ and $L$ are design choices selected to satisfy specific
performance objectives. For instance, choosing $J = 0$ and
$L = I$ yields a pure integral action, ensuring steady–state
tracking of the desired outputs. In the application section, we
illustrate how appropriate selections of $J$, $L$, and $C$
can be used to selectively target and attenuate specific
harmonic components. 
In harmonic variables, we have
\begin{equation}\label{eq:integral_state_harmo}
\dot Z(t)=\omega(t)\left((\hm{J}-\hm{N})Z(t)+\hm{LC}X(t)\right)\end{equation}
where the operators $\hm{J}=\hm{T}_{\theta}(J)$, $\hm{L}=\hm{T}_{\theta}(L)$, $\hm{C}=\hm{T}_{\theta}(C)$ are assumed constant and bounded on $\ell^2$. 
Stacking $\widetilde X=(X,Z)$ gives the augmented harmonic system
\begin{equation}\label{eq:augmented_harmonic_system0}
\dot{{\widetilde X}}(t)=\big(\widetilde{\hm{A}}(\omega(t))-\omega(t)\hm{N}\big){\widetilde X}(t)
+\widetilde{\hm{B}}(\omega(t))U(t),
\end{equation}
with 
\[
\widetilde{\hm{A}}(\omega)=
\begin{bmatrix}
\hm{G}(\omega)\hm{A}_0+\omega\hm{A}_1 & 0\\
\omega\hm{L}\hm{C} & \omega\hm{J}
\end{bmatrix},
\widetilde{\hm{B}}(\omega)=
\begin{bmatrix}
\hm{G}(\omega)\hm{B}_0+\omega\hm{B}_1\\
0
\end{bmatrix}.
\]
Let $U(t) = -\hm{K} {\widetilde X}(t) $ where $\hm{K}$ is a state-feedback operator to be designed. Applying Theorem~\ref{main} to the closed-loop operator
$\widetilde{\hm{A}}_{\mathrm{cl}}(\omega)=\big(\widetilde{\hm{A}}(\omega)-\omega\hm{N}\big)-\widetilde{\hm{B}}(\omega)\hm{K}$
yields, at each frequency vertex, a quadratic Lyapunov inequality in $\hm{P}$ and $\hm{K}$.
Introducing the standard change of variables $\hm{Y} = \hm{K}\hm{S}$, where $\hm Y$ and $\hm S=\hm P^{-1}$ are bounded Toeplitz block operators, yields the following set of linear matrix inequalities, which are convex in the decision variables: 
\begin{equation}\label{eq:LMPTh7}
(\widetilde{\mathcal{A}}_{\omega_i}-\omega_i\mathcal{N})\hm{S}
+\hm{S}{(\widetilde{\mathcal{A}}_{\omega_i}-\omega_i\mathcal{N})}^*
-\widetilde{\mathcal{B}}_{\omega_i}\mathcal{Y}
-\mathcal{Y}^*\widetilde{\mathcal{B}}_{\omega_i}^* \prec 0,
\end{equation}
and where $\widetilde{\mathcal{A}}_{\omega_i} = \widetilde{\mathcal{A}}(\omega_i)$, $\widetilde{\mathcal{B}}_{\omega_i} = \widetilde{\mathcal{B}}(\omega_i)$ with $\omega_i \in \{\omega_{\min},\omega_{\max}\}$ {and $\hm G(\omega_i)=\hm I$}. The LMIs~\eqref{eq:LMPTh7} are therefore a direct convex reformulation of the Lyapunov conditions of Theorem~\ref{main} applied to the augmented closed-loop AFM-LPP system. The feasibility of these LMIs guarantees closed-loop stability for all admissible time-varying frequencies {and provides a quadratic Lyapunov function given by $ {\widetilde X}^*{\hm P} {\widetilde X}$.} The resulting design yields a bounded Toeplitz state-feedback gain
$\hm K=\hm Y\hm S^{-1}=[\hm K_x,\hm K_z]$, admitting a time domain phase-periodic representation,
\begin{equation}
K(\theta(t))=\sum_{h\in\mathbb{Z}} K_h\,\mathrm{e}^{\mathrm{j}h\theta(t)} .\label{eqKperio}
\end{equation}
By introducing Toeplitz block weighting operators $\mathcal{R}$ and $\mathcal{Q}$, the LMIs~\eqref{eq:LMPTh7} can be interpreted in an LQR-like framework as classically done in the finite dimension LTI setting~\cite{boyd_linear_1994}. The resulting controller is a guaranteed-cost state-feedback law that ensures closed-loop stability under time-varying frequency and is obtained by solving the following optimization problem:
\begin{align}
&
 \min_{\scriptstyle \hm{S}, \hm{Y}, \mathcal{}} \mathrm{Tr}_0(\hm {M}), \text{ s.t. } \label{op2}\\
&\left(\begin{array}{ccc} \Xi(\omega_i)& \star & \star \\ \mathcal{R}^{\frac{1}{2}}\mathcal{Y} & -\mathcal{I} & \star \\\mathcal{Q}^{\frac{1}{2}}\hm{S}& 0 & -\mathcal{I}\end{array}\right)\preceq 0 \nonumber
\quad
\left(\begin{array}{cc}\hm {M}& \mathcal{I} \\\mathcal{I} & \hm{S}\end{array}\right)\succ0 \nonumber
\end{align}
with
\[
\Xi(\omega_i)=(\widetilde{\mathcal{A}}_{\omega_i}-\omega_i\mathcal{N})\hm{S}+\hm{S}{(\widetilde {\mathcal{A}}_{\omega_i}-\omega_i\mathcal{N})}^* -\widetilde{\mathcal{B}}_{\omega_i}\mathcal{Y}-\mathcal{Y}^*\widetilde{\mathcal{B}}_{\omega_i}^* 
\]
where $\Tr_{0}$ is the average-trace operator defined in~\cite{vernerey_tblmi_2025}, and $\hm {M}$ is a Toeplitz block operator bounded on $\ell^{2}$. 
\begin{remark}
 The stability and synthesis results established in this section are exact at the infinite-dimensional harmonic level and do not rely on frequency-rate bounds or operator approximation. In practice, numerical implementation requires finite-dimensional truncations of the Toeplitz operators, for which feasibility-preserving convergence properties are well known and do not alter the theoretical guarantees~\cite{vernerey_tblmi_2025}. The open-source \emph{PhasorArray Toolbox}~\cite{PhasorArray25,ECC26} automates the construction of truncated Toeplitz operators and provides routines for harmonic Lyapunov, Riccati, and LMI problems, making implementation straightforward. 
\end{remark}
\subsection{Classical harmonic controllers as special cases}
The harmonic-domain controller structure introduced above is designed to preserve the AFM-LPP property of the closed-loop system and thereby admit an exact Lyapunov-based analysis under variable-frequency conditions. In this section, it is shown that this formulation encompasses several classical harmonic control strategies, including proportional-resonant~\cite{teodorescu2006proportional, liserre_multiple_2006}, repetitive controllers~\cite{WEISS99} and modulation-based control strategies~\cite{Bodson} as special cases. This observation establishes a direct connection between the proposed approach and well-known designs, and provides further insight into its practical relevance. For instance, the ideal repetitive controller given by
$C(s) = \frac{e^{-sT_0}}{1 - e^{-sT_0}}$
possesses infinite gain at integer multiples of the fundamental 
frequency. Truncating its Fourier expansion yields a bank of resonant 
oscillators, each associated with a specific harmonic. Indeed, with $H$ the number of harmonics retained one has
$C(s) \approx \sum_{h = 1}^{H} 
\frac{K_h s}{s^2 + {(h\omega_0)}^2}$.
Each harmonic therefore introduces a resonant subsystem of the form:
\begin{equation}\label{t_osc}
\dot{z}_h =
\begin{bmatrix}
0 & h\omega_0\\
-h\omega_0 & 0
\end{bmatrix}
z_h +
\begin{bmatrix}
1\\
0
\end{bmatrix} e.
\end{equation}
Note that, under the rotating-frame change of variables $w_h = R(h\theta) z_h$, this formulation is equivalent to~\cite{Bodson}
\begin{equation}\label{eq:modulation}
\dot w_h =
\begin{bmatrix}
\cos(h\theta)\\
\sin(h\theta)
\end{bmatrix} e ,
\end{equation}
which corresponds to the integral structure used in modulation-based control schemes.
When $\omega(t)=\omega_0(t)$ is time-varying and the system dynamics are phase-periodic, the situation differs, as the targeted harmonics are defined in the phase domain rather than in time, leading to index-based harmonic resonance matrix:
\begin{equation}\label{t_osc_en_phase}
\pddphi{z_\theta}_{h}(\varphi) =
\begin{bmatrix}
0 & h\\
-h & 0
\end{bmatrix}
z_\theta (\varphi)+
\begin{bmatrix}
1\\
0
\end{bmatrix} e_\theta(\varphi).
\end{equation}
Consequently, within the proposed framework, the corresponding phase-domain realization is:
\begin{equation}\label{os_H}
\dot{Z}_h (t)=\omega(t)\left(
\begin{bmatrix}
-\hm{N} & h\hm{I}\\
-h\hm{I} & -\hm{N}
\end{bmatrix}
Z_h(t) +
\begin{bmatrix}
\hm{I}\\
0
\end{bmatrix} E(t)\right),
\end{equation}
which is directly compatible with the forwarding structure previously 
introduced. Stacking these oscillators for all selected harmonics 
produces a harmonic predictor that can be tuned via the same 
methodology used for the main controller. This viewpoint clarifies that PR and repetitive controllers are not 
separate design philosophies, but rather structured instances of the 
general internal-model construction enabled by the AFM-LPP harmonic 
representation. Moreover, it extends their applicability to systems 
with time-varying frequency.
\section{Application to Variable-Speed Motor Drives}\label{sec:pmsm_application}
To illustrate the relevance of the proposed AFM-LPP framework, we apply it to a permanent magnet synchronous motor (PMSM) operating under variable-speed conditions.  We show that the harmonic representation of the PMSM naturally fits into the AFM-LPP class introduced in Section~\ref{sec:robust_synthesis_LPP}, which enables the direct application of the exact stability and synthesis results developed in that section. This allows the construction of harmonic controllers ensuring tracking, harmonic mitigation, and closed-loop stability over a prescribed speed range. Throughout the remainder of this section, we assume that rotor speed $\omega_m(t)$ is \emph{measurable}. 

\subsection{PMSM as a phase-periodic system}
Consider a surface-mounted PMSM with $p$ pole pairs.
Let the state vector be $x=(i_{abc},\omega_m)$, where
$i_{abc}={(i_a,i_b,i_c)}^\top$ denotes the stator phase currents and
$\omega_m$ the rotor mechanical speed.
The control input is the applied stator voltage
$u=v_{abc}$, and the external disturbance is the load torque
$w=\Gamma_L$.
The PMSM dynamics can then be written in state-space form as
\begin{equation}\label{eq:pmsm_state_space}
\dot{x}(t) = A(\theta(t)) x(t) + B_u u(t) + B_w w(t),
\end{equation}
where the system matrices are given by
\begin{equation*}\label{eq:pmsm_matrices_tutorial}
\begin{aligned}
A(\theta) &= 
\begin{bmatrix}
-\frac{r}{L} I_3 & \frac{p\psi_f }{L} \Phi_{abc}(p\theta) \\[0.5em]
-\frac{p\psi_f}{J} {\Phi_{abc}(p\theta)}^\top & -\frac{B_f}{J}
\end{bmatrix}, \\[0.5em]
B_u &= 
\begin{bmatrix}
\frac{1}{L} I_3 \\[0.5em]
0_{1 \times 3}
\end{bmatrix}, \quad
B_w = 
\begin{bmatrix}
0_{3 \times 1} \\[0.5em]
-\frac{1}{J}
\end{bmatrix}.
\end{aligned}
\end{equation*}
The machine parameters are the stator phase inductance $L$,
phase resistance $r$, rotor inertia $J$,
viscous friction coefficient $B_f$,
and permanent magnet flux linkage $\psi_f$.
The vector $\Phi_{abc}(p\theta)$ represents the back-EMF waveform,
which is $2\pi$-periodic in the electrical angle $p\theta$.
For an ideal surface-mounted PMSM, this waveform is purely sinusoidal and given by
\begin{equation}\label{eq:back_emf}
\Phi_{abc}(p\theta) = \begin{bmatrix}
\sin(p\theta) \\
\sin(p\theta - 2\pi/3) \\
\sin(p\theta + 2\pi/3)
\end{bmatrix}
\end{equation}
 {The state matrix $A(\theta)$ depends periodically on the electrical angle through $\Phi_{abc}(p\theta)$, which makes the PMSM dynamics {phase-periodic}. Consequently, the PMSM provides a canonical example of an AFM-LPP system and is therefore directly compatible with the harmonic modeling and control framework developed in Section~\ref{sec:robust_synthesis_LPP}. Applying the variable-frequency Fourier decomposition $\hm{F}_{\theta}$ to the state-space model~\eqref{eq:pmsm_state_space}, where the phase variable $\theta$ satisfies $\dot{\theta}=\omega$ with $\omega$ a time-varying parameter, leads to the following AFM-LPP harmonic representation which matches the exact harmonic PMSM model when $\omega=\omega_m$:}


\begin{equation}\label{eq:pmsm_harmonic_lpv}
\begin{aligned}
\dot{X}(t) =& \big(\hm{G}(\omega(t))\hm{A} - \omega(t)\hm{N}\big)X(t) \\&+ \hm{G}(\omega(t))(\hm{B}_u U(t) + \hm{B}_w W(t))
\end{aligned}
\end{equation}
where $X = \hm{F}_{\theta}(x) = (I_{abc}, \Omega)$, 
$U = V_{abc}=\hm{F}_{\theta}(u)$ and $W = \hm{F}_{\theta}(w)$ 
and where the constant Toeplitz-block operators $\hm{A}$, $\hm{B}_u$, and $\hm{B}_w$ are given by:
\begin{subequations}\label{eq:harmonic_operators_pmsm}
\begin{align}
&\hm{A} = \begin{bmatrix}
-\frac{r}{L} \hm{I}_3 & \frac{\psi_f p}{L} \hm{T}_{\theta}(\Phi_\mathrm{abc}) \\[0.5em]
-\frac{p\psi_f}{J} {\hm{T}_{\theta}(\Phi_\mathrm{abc})}^* & -\frac{B_f}{J}\hm{I}_1
\end{bmatrix}, \\ \qquad
&\qquad\hm{B}_u = \begin{bmatrix}
\frac{1}{L}\hm{I}_3 \\[0.5em]
0_{1 \times 3}
\end{bmatrix}, \qquad
\hm{B}_w = \begin{bmatrix}
0_{3 \times 1} \\[0.5em]
-\frac{1}{J}\hm{I}_1
\end{bmatrix}
\end{align}
\end{subequations}
where $\hm{I}_n = \hm{T}_{\theta}(\idn)$, and  $\hm{T}_{\theta}(\Phi_\mathrm{abc})$ represents the back-EMF harmonics. For the ideal sinusoidal back-EMF~\eqref{eq:back_emf}, the Toeplitz operators $\hm{T}_{\theta}(\Phi_{abc})$ exhibit a sparse structure. 
For example, since
$\Phi_a(\theta)
 = \sin(p\theta)
 = \tfrac{1}{2\ji}\bigl(e^{j p\theta} - e^{-j p\theta}\bigr)$,
the Toeplitz matrix $\hm{T}_{\theta}(\Phi_a)$ has only two nonzero
diagonals, at offsets $\pm p$, and each of their diagonal elements are equal to
$\pm\tfrac{ \ji}{2}$. The harmonic model~\eqref{eq:pmsm_harmonic_lpv} is affine in the frequency $\omega$ and Toeplitz structured, and therefore satisfies the AFM-LPP assumptions required by the analysis and synthesis results of Section~\ref{control_design}.
\subsection{Control objectives and synthesis}
The control design aims to enforce phase-periodic steady-state behavior and harmonic mitigation, while preserving stability under variable-frequency operation. In particular, the following objectives are considered:
\begin{enumerate}
 \item \emph{Speed regulation:} accurate tracking of the mechanical speed $x_4(t)=\omega_m(t)$ to a given reference $\omega_{m,0}^{\mathrm{ref}}$ in the presence of load torque disturbances $\Gamma_L$. 
 \item \emph{Power factor optimization:} regulation of the $d$-axis current to zero (or to a prescribed reference) in order to maximize efficiency.
 \item \textit{Harmonic mitigation:} rejection of phase-periodic current harmonics induced by torque ripple at $2\omega_m$, which generates current harmonics in $i_{abc}$ of order $p \pm 2k$, $k = 1,2,\ldots$.
\end{enumerate}
These objectives are addressed by enforcing zero steady-state values of selected phase-domain output components using integral action in the harmonic domain. Speed regulation is achieved by integrating the error between the DC component $\Omega_0(t)$ of the mechanical speed phasor and the reference $\omega_{m,0}^{\mathrm{ref}}$. 

Power factor optimization is achieved by regulating the DC component of the $d$-axis current $i_d$ in the rotating $dq$ frame where $i_{dq}=T_p(\theta)(i_{abc})$ and $v_{dq}=T_p(\theta)(v_{abc})$ and where the generalized Park transformation~\cite{MIRHO} is provided by (here, for $k=p$)
\begin{align} 
i_{\mathrm{dq}}^k &= T_k(\theta) i_{abc} \label{park}\\
&= \tfrac{2}{3}
\begin{bmatrix}
\cos(k\theta) & \cos(k\theta - \tfrac{2\pi}{3}) & \cos(k\theta + \tfrac{2\pi}{3}) \\
-\sin(k\theta) & -\sin(k\theta - \tfrac{2\pi}{3}) & -\sin(k\theta + \tfrac{2\pi}{3})
\end{bmatrix} i_{abc}.\nonumber
\end{align}
This corresponds to enforcing $I_{d,0}=0$, where $I_d=\hm F_{\theta}(i_d)$ is the harmonic representation of the $d$-axis current.

Harmonic mitigation is implemented by exploiting generalized Park transformations $T_k(\theta)$ for $k\neq p$
which map the $k$-th harmonic components of the three-phase currents into DC signals in the rotating frame. Integral action on these quantities enforces rejection of the corresponding harmonics in phase-periodic steady state. The resulting control output operator is defined as
\begin{equation}\label{eq:output_matrix_time}
C(\theta)=
\begin{bmatrix}
0 & 1 \\
[1\ 0]T_p(\theta) & 0 \\
\mathrm{stack}_{k=0,2,6,8}\!\left(T_k(\theta)\ \ 0\right)
\end{bmatrix},
\end{equation}
which selects the relevant speed, $d$-axis, and harmonic components to be regulated.

In the harmonic-domain framework, tracking of phase-periodic operating trajectories requires shifting the origin to the harmonic equilibrium associated with the reference speed and disturbance. The phase-periodic steady state dynamics can be characterized from the definition of an equilibrium in phase domain:
\begin{definition} In phase-domain, an equilibrium $(X^{\mathrm{ref}},U^{\mathrm{ref}})$ for a phase-periodic disturbance $W$ satisfies the non-linear equation: 
\begin{equation}\label{eq:harmonic_equilibrium1}
0=\big(\hm A-\hm T_\theta(\omega_m^{\mathrm{ref}})\hm N\big)X^{\mathrm{ref}}
+\hm B_u U^{\mathrm{ref}}+\hm B_u W,
\end{equation}
where the phase frequency $\omega_m^{\mathrm{ref}}$ corresponds to the speed rotor reference.
\end{definition}
Assuming a known phase-periodic disturbance $W$ and the reference speed trajectory $\omega_{m}^{\mathrm{ref}}(t)$ takes values in the interval $[\omega_{\min},\ \omega_{\max}]$, the full derivation of
the equilibrium $(X^{\mathrm{ref}},U^{\mathrm{ref}})=\hm F_{\theta}(x^{\mathrm{ref}},u^{\mathrm{ref}})$ is given in Appendix.
Note that, when the phase function $\omega$ differs from the reference $\omega_m^{\mathrm{ref}}$, the corresponding equilibrium quantities evolve according to the harmonic dynamics associated with the actual phase trajectory
\begin{equation}\label{eq:harmonic_equilibrium}
  \begin{split}
      \dot X^{\mathrm{ref}}(t)
    =&\big(\hm G(\omega(t))\hm A-\omega(t)\hm N\big)X^{\mathrm{ref}}(t)\\
    &+\hm G(\omega(t))\big(\hm B_u U^{\mathrm{ref}}(t) +\hm B_w W\big).
  \end{split}
\end{equation}
We now rewrite the system around the phase-periodic equilibrium in order to synthesize a stabilizing controller for the deviation dynamics. Let the state and input errors be given by
$E=X-X^{\mathrm{ref}}$ and $\widetilde U=U-U^{\mathrm{ref}}$,
where $(X^{\mathrm{ref}},U^{\mathrm{ref}})$ satisfies the harmonic equilibrium
equation~\eqref{eq:harmonic_equilibrium}.
The resulting error dynamics are given by
\begin{align*}
 \dot E(t)=(&\hm{G}(\omega(t))\hm A -\omega(t)\hm N)E(t)+\hm{G}(\omega(t))\hm B_u \widetilde U(t)
\end{align*}
Using the control output as defined by~\eqref{eq:output_matrix_time}, the harmonic forwarding dynamics are introduced according to~\eqref{eq:integral_state_harmo} by choosing $\hm J=0$ and $\hm L=\hm I$, yielding
\begin{equation}\label{eq:integral_dynamics_harmonic}
\dot{Z}(t)=\omega(t)\big(-\hm N Z(t)+\hm CE(t)\big).
\end{equation}
The corresponding harmonic-augmented system then reads with ${\widetilde E}=(E,Z)$
\begin{equation}\label{eq:augmented_harmonic_system}
\begin{aligned}
\dot{{\widetilde E}}(t)=&
\big(\widetilde{\hm A}(\omega(t))-\omega(t)\hm N\big){\widetilde E}(t)
+\widetilde{\hm B}_v(\omega(t)) \widetilde U(t)
\end{aligned}
\end{equation}
where 
\begin{align}\label{eq:augmented_matrices}
&\widetilde{\hm A}(\omega) =
\begin{bmatrix}
\hm{G}(\omega)\hm A & 0 \\
\omega\hm C & 0
\end{bmatrix}, \quad 
\widetilde{\hm B}_v(\omega) =
\begin{bmatrix}
\hm{G}(\omega)\hm B_u \\ 0
\end{bmatrix}.
\end{align}
  
The augmented deviation system~\eqref{eq:augmented_harmonic_system} is an AFM-LPP system. The synthesis results of Section~\ref{sec:robust_synthesis_LPP} therefore yield a Toeplitz-block feedback gain $\hm K=[\hm K_x\ \hm K_z]$ obtained by solving the harmonic synthesis problem~\eqref{op2} and using the reconstruction procedure~\eqref{recons_timedomain}. The corresponding time-domain control law is
\begin{equation}
  \begin{split}
    u(t)=&\;K_x(\theta(t))\big(x(t)-x^{\mathrm{ref}}(t)\big) \\
    &+K_z(\theta(t))z(t)+u^{\mathrm{ref}}(t).
  \end{split}
\end{equation} 
This controller guarantees asymptotic convergence $E(t)=X(t)-X^{\mathrm{ref}}(t)\to0$ provided that $\omega \in [\omega_{\min},\omega_{\max}]$. However, for the PMSM, $\omega$ is not a free variable but is constrained by the mechanical speed, i.e., $\omega=\omega_m$ and $\theta=\theta_m$. Therefore, the parameter $\omega$  coincides with a state variable, and the stability guarantee holds as long as the state trajectory remains inside the admissible domain $[\omega_{\min},\omega_{\max}]$. 
Proposition~\ref{prop1_appendix}, given and proved in the appendix, provides a sufficient condition ensuring this property.
This result ensures that the frequency admissibility condition required by the AFM-LPP stability theorem is preserved along closed-loop trajectories. It is obtained by determining the largest quadratic Lyapunov level set contained inside the admissible domain. The admissible level is uniquely determined by the minimum distance between the DC component $\omega_{m,0}^{\mathrm{ref}}$ of the rotor speed reference and the bounds $\omega_{\min}$ and $\omega_{\max}$. More specifically, let the Lyapunov function
$ \hm V({\widetilde E}) = {\widetilde E}^{*}\, \hm P {\widetilde E}$,
where ${\widetilde E}$ is the augmented state defined in~\eqref{eq:augmented_harmonic_system} and $\hm P=\hm S^{-1}$, where $\hm S$ is the solution of~\eqref{op2}. Let $S$ be the representative of $\hm S$ in time domain and denote by $s_{44}$ the fourth diagonal entry of $S$ and by ${S}_{44,0}$ the zero-order phasor of ${S}_{44}=\hm F_\theta(s_{44})$.
For a given step reference $\omega_{m,0}^{\mathrm{ref}}(t_0)$, if the initial augmented state ${\widetilde E}(t_0)$ satisfies
\begin{equation}
\hm V({\widetilde E}(t_0)) < L_{\max}(\omega_{m,0}^{\mathrm{ref}}(t_0)),\label{decision0}
\end{equation}
where
\[
L_{\max}(\omega_{m,0}^{\mathrm{ref}})
=
\frac{\delta^2(\omega_{m,0}^{\mathrm{ref}})}{ S_{44,0}},
\]
with 
$$\delta(\omega_{m,0}^{\mathrm{ref}})
=
\min\big(
|\omega_{\min}-\omega_{m,0}^{\mathrm{ref}}|,
|\omega_{\max}-\omega_{m,0}^{\mathrm{ref}}|
\big), $$
then the closed-loop trajectory satisfies
\[
E(t)\to0
\quad\text{and}\quad
\omega_m(t)\in[\omega_{\min},\omega_{\max}]
\]
for all $t\ge t_0$. Moreover, any subsequent step change of the reference $\omega_{m,0}^{\mathrm{ref}}(t_k)$ $k=1,2,\cdots $ preserving~\eqref{decision0} is admissible and does not affect stability.
\color{black}
\subsection{Simulation results}\label{sec:simulation_results}
To validate the proposed harmonic control framework, numerical simulations are performed on the PMSM operating under variable-speed conditions and phase-periodic load torque disturbances.
The harmonic synthesis problem~\eqref{op2} is solved using the weighting matrices
$Q=\hm I$ and $R=100\hm I$. The operating speed range is chosen as
$[\omega_{\min},\omega_{\max}]=[10,200]$~rad/s (approximately 100–2000~rpm), representative of low-to-medium speed operation.
Using a $10$th-order truncation scheme, the computation time required to obtain a sufficiently accurate solution is $114$~s. The PMSM parameters used in simulation are reported in Table~\ref{tab:pmsm_typical_params}.
\begin{table}[htbp]
\centering
\caption{PMSM parameters}\label{tab:pmsm_typical_params}
\begin{tabular}{l c c}
\toprule
\textbf{Parameter} & \textbf{Symbol} & \textbf{Value} \\
\midrule
Stator resistance  & $r$  & $0.5~\Omega$ \\
Phase inductance   & $L$  & $1.5~\mathrm{mH}$ \\
Flux linkage    & $\psi_f$ & $0.14~\mathrm{Wb}$ \\
Rotor inertia   & $J$  & $0.03~\mathrm{kg\,m^2}$ \\
Viscous friction coefficient & $B_f$  & $0.02~\mathrm{N\,m\,s/rad}$ \\
Number of pole pairs  & $p$  & $4$ \\
\bottomrule
\end{tabular}
\end{table}
\begin{figure*}[t]
 \centering
 \begin{minipage}{0.45\textwidth}
  \centering
 \includegraphics[width=\linewidth]{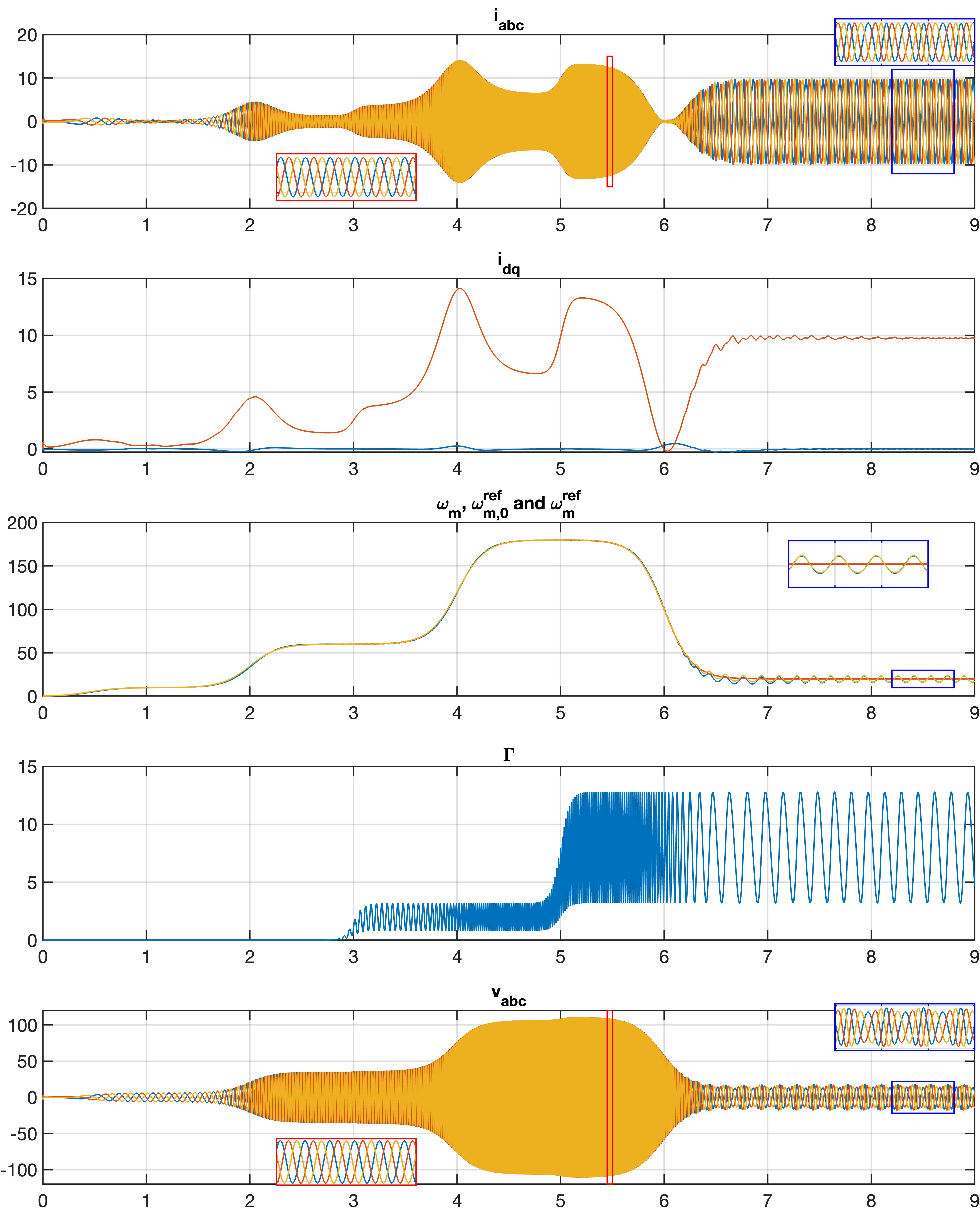}
 \caption{Phase currents $i_\mathrm{abc}$, current $i_{dq}$, rotor speed $\omega_m$ (blue), $\omega_m^{\mathrm{ref}}$ (yellow) and $\omega_{m,0}^{\mathrm{ref}}$ (red), load torque $\Gamma_L$ and the applied control $u=v_\mathrm{abc}$.}\label{fig:simul}
 \end{minipage}\hfill
 \begin{minipage}{0.45\textwidth}
   \centering
 \includegraphics[width=\linewidth]{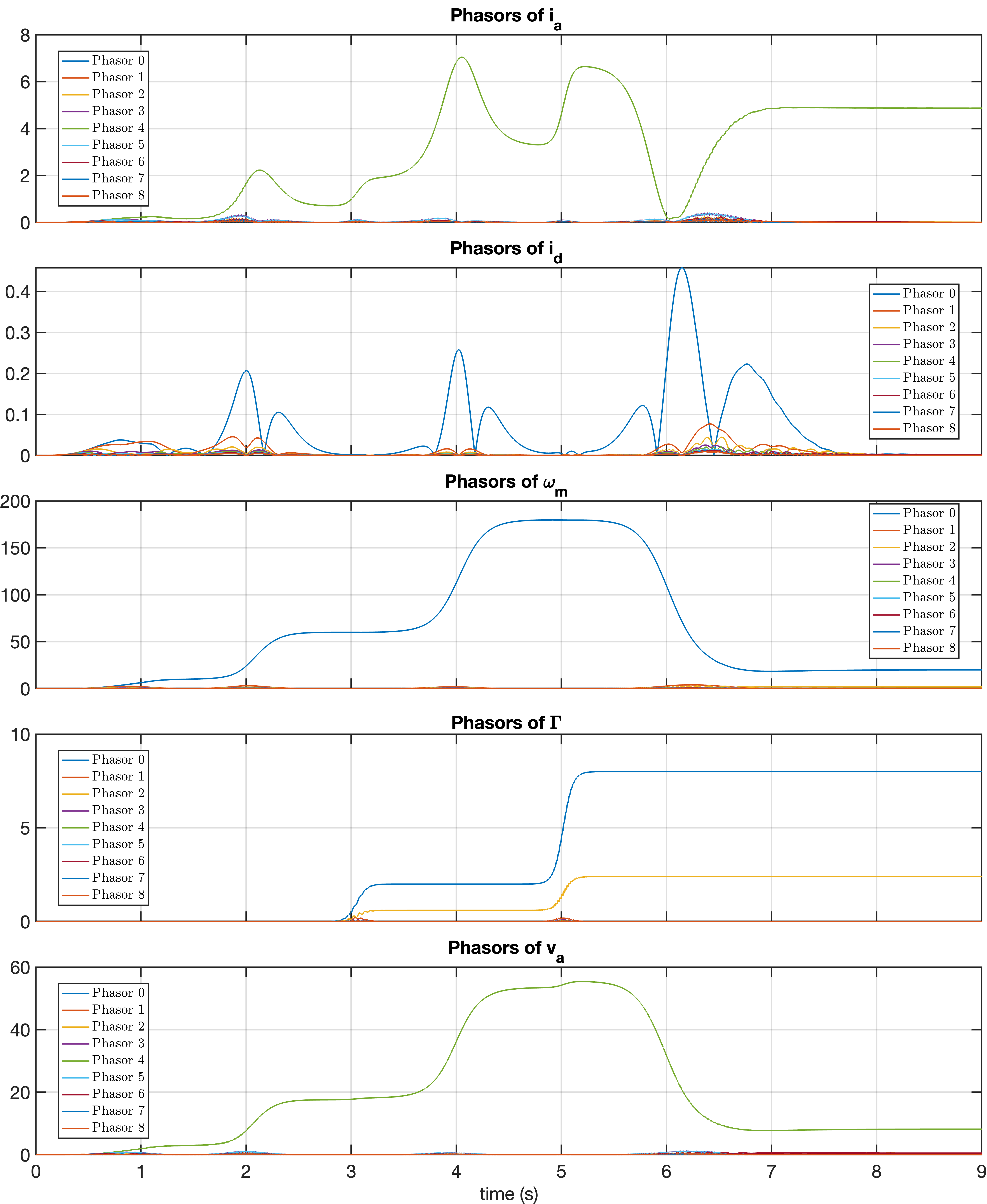}
 \caption{ {Magnitude of harmonics illustrating targeted suppression in $i_a$ and $i_d$. Top to bottom: $ I_{a,k}$, $I_{d,k}$, $\Omega_{m,k}$ $\Gamma_k$ and $V_{a,k}$, $k=0,1,\ldots,8$.}}\label{fig:simulharmo}
 \end{minipage}
\end{figure*}
\begin{figure*}[t]
\begin{minipage}{0.45\textwidth}
  \centering
 \includegraphics[width=\linewidth]{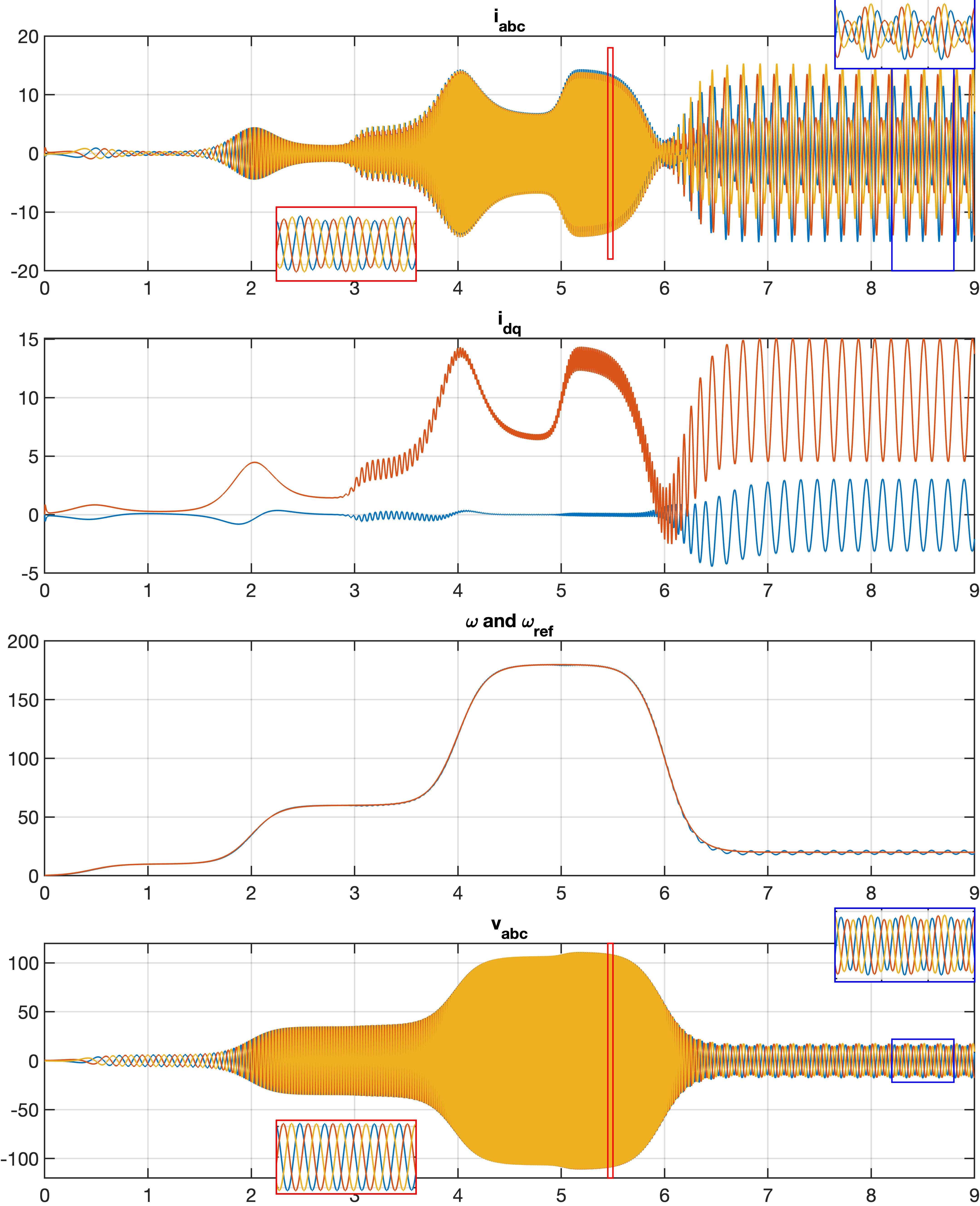}
 \caption{Simulation results without harmonic mitigation: the phase currents $ i_{\mathrm{abc}} $, the $dq$ currents $ i_{dq} $, the rotor speed $ \omega_m $ and the applied control $v_{abc}$, for the same load torque and reference speed profile as in the previous case.}\label{fig:simulsans}
 \end{minipage}\hfill
 \begin{minipage}{0.45\textwidth}
  \centering
 \includegraphics[width=\linewidth]{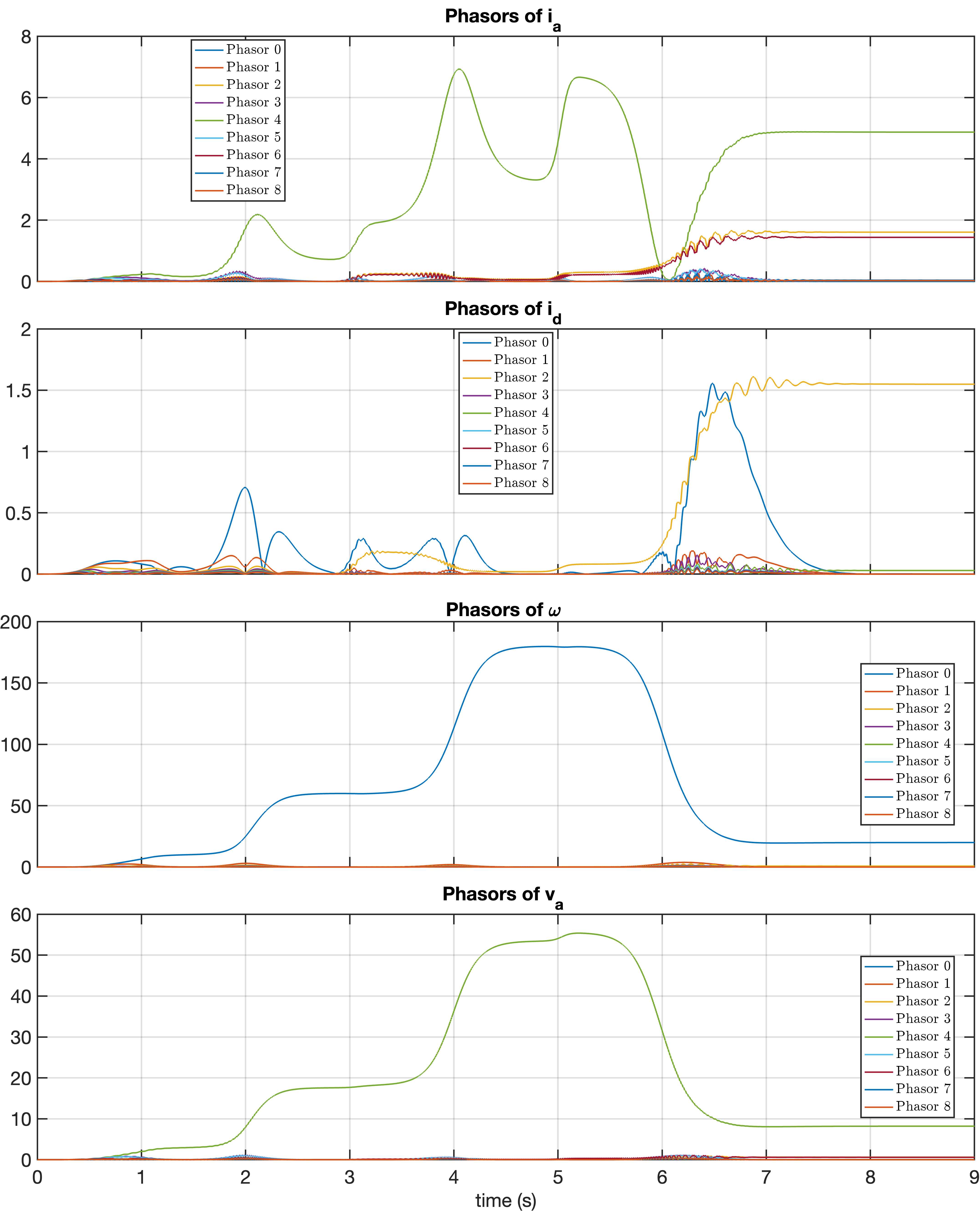}
 \caption{ {Magnitude of harmonics without harmonic mitigation objective. Top to bottom: $ I_{a,k}$, $I_{d,k}$, $\Omega_{m,k}$ and $V_{a,k}$, $k=0,1,\ldots,8$.}}\label{fig:simulharmosans}
 \end{minipage}
\end{figure*}

Figure~\ref{fig:simul} shows the closed-loop time-domain responses obtained with the synthesized harmonic controller, including the phase currents $i_{abc}$, $dq$-currents, rotor speed $\omega_m$, load torque $\Gamma_L$ and the applied control voltages $v_{abc}$. The corresponding harmonic spectra are displayed in Figure~\ref{fig:simulharmo}. For comparison, Figures~\ref{fig:simulsans} and~\ref{fig:simulharmosans} report the system response obtained when harmonic mitigation is not included in the control objectives, while keeping identical speed references and torque disturbances.
The simulation results highlight the following key properties:
\begin{itemize}
 \item \emph{Speed regulation:} The rotor speed $\omega_m$ accurately tracks its time-varying reference $\omega_{m}^{\mathrm{ref}}$, defined in~\eqref{wref} (Appendix) and computed from $\omega_{m,0}^{\mathrm{ref}}$, over the entire operating range, as illustrated in the zoom of Figure~\ref{fig:simul}. This supports the mechanical filtering assumption underlying its construction.
 \item \emph{$d$-axis current regulation:} the mean value of the direct-axis current $i_d$ is maintained close to zero, ensuring power factor optimization.
 \item \emph{Harmonic mitigation:} when harmonic rejection is enabled, the steady-state phase currents $i_{abc}$ exhibit no harmonic components other than the fundamental $p$-th harmonic, independently of the reference speed, as shown in Figure~\ref{fig:simulharmo}. In particular, this implies $i_d = 0$ and $i_q = \mathrm{const}$ in phase-periodic steady state. 
\end{itemize}
 Figure~\ref{fig:simul} includes two zoomed views: one at quasi-constant rotor speed and one under time-varying speed. These plots highlight that phase-periodicity with respect to the electrical angle $\theta$ should not be confused with periodicity in the time domain, as the latter is generally non-sinusoidal under variable-speed operation (see the zoomed area around $t=8$~s).
 {Finally, the level-set guarantee is illustrated in Figure~\ref{fig:levelset}, which compares the theoretical maximum admissible Lyapunov level $L_{\max}$ with the closed-loop Lyapunov trajectory. The results show that the Lyapunov function remains strictly below the admissible bound during reference step changes, confirming closed-loop stability is preserved.
}
 \begin{figure}
  \centering
  \includegraphics[width=0.9\linewidth]{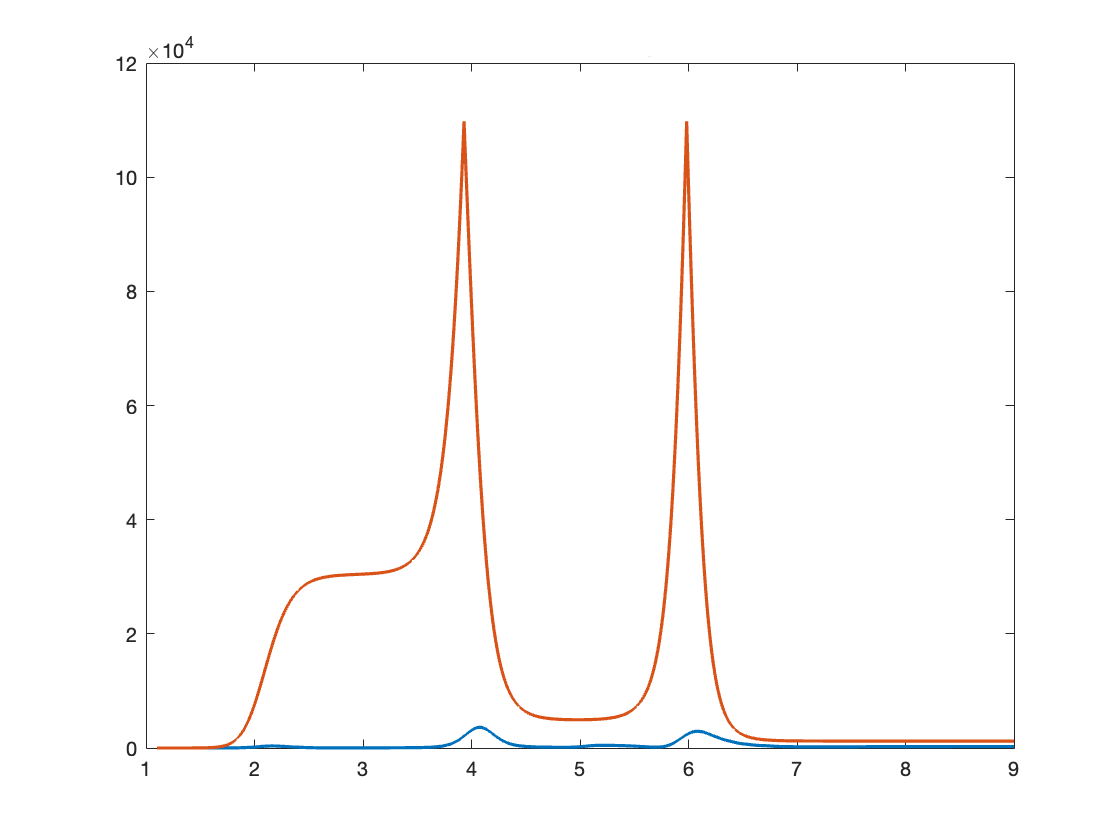}
  \caption{Admissible level $L_{\max}$ (red) and Lyapunov function $\hm V({\widetilde E})$ (blue) for the reference $\omega_{m,0}^{\mathrm{ref}}$ of Fig.~\ref{fig:simul}.}\label{fig:levelset}
 \end{figure}
\section*{Conclusion}\label{sec:conclusion}
\vspace{-.3cm}
This paper establishes a unified and rigorous framework for harmonic modeling and control of systems under variable-frequency. An exact harmonic-domain formulation is derived, providing a principled basis for analysis and systematic approximation, together with a quantitative validity criterion that characterizes the modeling error induced by frequency variations.
For linear phase-periodic systems with affine dependence on the instantaneous frequency, stability and state-feedback synthesis are shown to be certified exactly by enforcing Lyapunov inequalities at a finite number of frequency vertices. This result bridges harmonic modeling and LMI-based control design, enabling tractable certification over continuous frequency intervals.
The approach is illustrated through harmonic mitigation in variable-speed electrical drives, demonstrating both practical relevance and effectiveness for systems operating under nonstationary frequency conditions. Beyond electrical applications, the proposed framework provides a general harmonic-domain methodology for stability analysis and control of phase-periodic systems with time-varying frequency.

\section{Appendix}
\vspace{-.3cm}
This appendix provides technical details supporting the application section. 
First, the phase–periodic equilibrium of the PMSM is derived in the harmonic domain. 
Then, the maximum admissible Lyapunov level ensuring frequency admissibility is characterized.

\subsection{PMSM Phase-periodic equilibrium}
\vspace{-.3cm}
Given a known phase-periodic disturbance $W$, the harmonic equilibrium 
$(X^{\mathrm{ref}},U^{\mathrm{ref}})=\hm F_{\theta}(x^{\mathrm{ref}},u^{\mathrm{ref}})$ 
is obtained from the $dq$-frame PMSM model derived using the Park transform $T_p$ defined in~\eqref{park}. 
The corresponding phase-periodic harmonic model is given by
\begin{align*}\label{eq:pmsm_dq_harmonic}
 L\dot I_{d}=&\hm{G}(\omega(t))(-rI_d+V_d)+\omega(t)L(p I_q- \hm{N}I_d)\\
L\dot I_{q}=&\hm{G}(\omega(t))\big(-rI_q+V_q\big)-\omega(t)(p\psi\hm I+pLI_d+L\hm{N}I_q)\\ 
J\dot\Omega=&\hm{G}(\omega(t))(\frac{3}{2}p\psi I_q-W-b_f\Omega)-\omega(t)J\hm{N}\Omega
\end{align*}
where $I_{dq}=\hm{F}_\theta(i_{dq})$ and $V_{dq}=\hm{F}_\theta(v_{dq})$ denote the harmonic representations of the $dq$ currents and voltages.

Since the physical PMSM satisfies $\omega=\omega_m$, and under phase-periodic steady-state operation (i.e., when $\Omega = \hm F_\theta(\omega_m)$ is constant), the equilibrium equation is obtained by left-multiplying by $\hm{G}^{-1}(\omega)$:
\begin{align*}
 0=&\frac{1}{L}(-rI_d+V_d)+\hm T_\theta(\omega_m)(pI_q- \hm{N}I_d)\\
0=&\frac{1}{L}(-rI_q+V_q)-\hm T_\theta(\omega_m)(\frac{p\psi}{L}\hm{I}+pI_d+\hm{N}I_q)\\ 
0=&\frac{1}{J}(\frac{3}{2}p\psi I_q-W-b_f\Omega)-\hm T_\theta(\omega_m)\hm{N}\Omega
\end{align*}
Under the selected integral actions, $\Omega_0 = \omega_{m,0}^{\mathrm{ref}}$, $I_{d,k} = 0$ for all $k \in \mathbb{Z}$, and $I_{q,k} = 0$ for all $k \neq 0$, 
the $k$-th harmonic component of the last equation with $k\neq 0$ yields the nonlinear harmonic recursion
\begin{equation}
 \Omega_k=H(\ji k\omega_{m,0}^{\mathrm{ref}})\big( \frac{1}{J}W_k+\sum_{p=-\infty,p\neq k}^{+\infty}\Omega_{k-p}\Omega_{p}(\ji p)\big)\label{rec} 
\end{equation}
where
\[
H(\ji k\omega)=\frac{-1}{\frac{b_f}{J}+\mathrm{j}k\omega}.
\]
Equation~\eqref{rec} defines a nonlinear harmonic fixed-point relation. The equilibrium can therefore be computed using a fixed-point iteration initialized from the fundamental component.
For $k=0$, using the symmetry property $\Omega_{-p}=\Omega_p^*$, the classical steady-state relation is recovered:
\[
I_{q,0}=\frac{2}{3p\psi}\big(W_0+b_f\omega_{m,0}^{\mathrm{ref}}\big),
\qquad
\Omega_0=\omega_{m,0}^{\mathrm{ref}}.
\]
In particular, if $W_k=0$ for all $k\neq0$, then $\Omega_k=0$ for all $k\neq0$. The recursion~\eqref{rec} reveals that higher-order mechanical harmonics arise from nonlinear harmonic coupling. However, due to the low-pass filtering effect of the mechanical dynamics, higher-order components are strongly attenuated, which justifies truncating the recursion in practice. As an illustration, consider a disturbance containing only a second harmonic component $W_2$. Starting from the initialization $\Omega_k=0$ for all $k\neq0$, the first iteration yields
\[
\Omega_{2}
=
H(2\mathrm{j}\omega_{m,0}^{\mathrm{ref}})\frac{
1}{J}W_2 \text{ and } \Omega_{-2}=\bar \Omega_2
\]
Higher-order components are generated through harmonic coupling, for example, at second iteration appears: 
\[
\Omega_{4}
=
H(4\mathrm{j}\omega_{m,0}^{\mathrm{ref}})\Omega_2^2(2\mathrm{j}).
\]
Because of mechanical filtering, these higher-order terms are typically negligible, leading to the approximation $\Omega_{\pm4}\approx0$. Under this approximation, the diffusion effect vanishes and the harmonic equilibrium is given by
\begin{align*}
 I_{q,0}^{\mathrm{ref}} &= \frac{2}{3p\psi_f}\bigl(W_0 + B_f \omega_{m,0}^{\mathrm{ref}}\bigr), \\
 \Omega_0^{\mathrm{ref}} &= \omega_{m,0}^{\mathrm{ref}}, \\
 \Omega_2^{\mathrm{ref}} &= \frac{1}{J} H(2 \ji \omega_{m,0}^{\mathrm{ref}})\, W_2, \\
 \Omega_k^{\mathrm{ref}} &= 0, \quad k \neq 0,\pm 2,
\end{align*}
These harmonic-domain equilibrium quantities uniquely define the corresponding phase-periodic time-domain reference trajectories:
\begin{subequations}
\begin{align}
 i_d^{\mathrm{ref}}(t) &= 0, \\
 i_q^{\mathrm{ref}}(t) &= \frac{2}{3p\psi_f}\bigl(W_0 + B_f \omega_{m,0}^{\mathrm{ref}}\bigr), \\
 \omega_m^{\mathrm{ref}}(t) &= \omega_{m,0}^{\mathrm{ref}}
 + 2 |\Omega_2| \cos\!\bigl(2\theta(t) - \varphi\bigr),\label{wref}
\end{align}
\end{subequations}
where $ \varphi = \angle \Omega_2 $.
The control reference $U^{\mathrm{ref}}$ is then obtained from the electrical equilibrium equations:
\begin{align*}
 V_d^{\mathrm{ref}} &= -\hm T_\theta(\omega_m^{\mathrm{ref}})\, pL\, I_q^{\mathrm{ref}}, \qquad 
 V_q^{\mathrm{ref}} &= r I_q^{\mathrm{ref}}
 + \hm T_\theta(\omega_m^{\mathrm{ref}})\, p\psi_f,
\end{align*}
or, equivalently, in the time domain,
\begin{subequations}
\begin{align*}
 v_d^{\mathrm{ref}}(t) &= -\omega_m^{\mathrm{ref}}(t)\, pL\, i_q^{\mathrm{ref}}(t), \\
 v_q^{\mathrm{ref}}(t) &= r\, i_q^{\mathrm{ref}}(t)
 + \omega_m^{\mathrm{ref}}(t)\, p\psi_f.
\end{align*}
\end{subequations}
Finally, the corresponding phase-periodic equilibrium $(x^{\mathrm{ref}},u^{\mathrm{ref}})$ is obtained via the inverse Park transformation,
$i_{abc}^{\mathrm{ref}}=T_p^{-1}(i_{dq}^{\mathrm{ref}})$ and
$v_{abc}^{\mathrm{ref}}=T_p^{-1}(v_{dq}^{\mathrm{ref}})$. 
\subsection{Maximum Admissible Lyapunov Level}
The following proposition provides the largest Lyapunov level set ensuring that the rotor speed trajectory remains inside the admissible interval $[\omega_{\min}, \omega_{\max}]$. Geometrically, this corresponds to determining the largest ellipsoidal level set of the Lyapunov function that is fully contained inside the admissible speed domain. Since the Lyapunov function is quadratic and radially unbounded, the maximal admissible level is necessarily attained when the speed trajectory reaches one of the admissible boundaries. The proof therefore reduces to characterizing the Lyapunov level associated with these boundary conditions.
\begin{proposition}\label{prop1_appendix}
Consider, at time $t_0$, a speed reference $\omega_{m,0}^{\mathrm{ref}}(t_0)$ and the corresponding equilibrium pair $(X^{\mathrm{ref}}, U^{\mathrm{ref}})$. Define the Lyapunov function
\[
\hm V({\widetilde E}) = {\widetilde E}^{*}\, \hm P {\widetilde E},
\]
where ${\widetilde E}$ is the augmented state defined in~\eqref{eq:augmented_harmonic_system}, and $\hm P=\hm S^{-1}$ with $\hm S$ is the solution of problem~\eqref{op2}.  
If the initial condition satisfies
\begin{equation}
\hm V({\widetilde E}(t_0)) < L_{\max}(\omega_{m,0}^{\mathrm{ref}}),\label{decision}
\end{equation}
where $L_{\max}$ denotes the largest Lyapunov level set for which the rotor speed component $x_4=\omega_m$ remains inside the admissible interval $ \omega_m(t) \in [\omega_{\min}, \, \omega_{\max}]$,
then the error state satisfies $E(t) \to 0$ as $t \to \infty$. Furthermore, any subsequent step change in the reference $\omega_{m,0}^{\mathrm{ref}}(t_k)$, $k=1,2,\ldots$, that preserves condition~\eqref{decision} is admissible and does not affect stability. The admissible level $L_{\max}$ is given by
\begin{equation}
L_{\max}(\omega_{m,0}^{\mathrm{ref}})
=
\frac{\delta^2(\omega_{m,0}^{\mathrm{ref}})}
{{S}_{44,0}},\label{eqLmax}
\end{equation}
where $\delta(\omega_{m,0}^{\mathrm{ref}})
=
\min\!\bigl(
|\omega_{\min}-\omega_{m,0}^{\mathrm{ref}}|,
|\omega_{\max}-\omega_{m,0}^{\mathrm{ref}}|
\bigr), $
and ${S}_{44,0}$ denotes the zero-order phasor of ${S}_{44} = \hm F_\theta(s_{44})$, with $s_{44}$ being the fourth diagonal entry of the time-domain realization $S$.
\end{proposition}
\vspace{-.5cm}
\begin{pf}
Consider a reference $\omega_{m,0}^{\mathrm{ref}}\in[\omega_{\min}, \omega_{\max}]$  and the associated equilibrium $(X^{\mathrm{ref}}, U^{\mathrm{ref}})$. Assume that the corresponding rotor speed reference $\omega_m^{\mathrm{ref}}$ remains inside the admissible interval $[\omega_{\min}, \omega_{\max}]$. The objective is to determine the largest Lyapunov level set $L_{\max}$ for which the rotor speed component $x_4=\omega_m$ remains inside this interval.

Since $\omega(t)>0$ if and only if $\hm T_\theta(\omega)(t)\succ0$, admissibility of the rotor speed is equivalent to the harmonic constraints
\[
\omega_{\min}\hm I \preceq \hm T_\theta(\omega_m)(t) \preceq \omega_{\max}\hm I.
\]
Because the Lyapunov function is radially unbounded, the maximal admissible level set is necessarily attained on the boundary of these constraints. Considering, for instance, the lower bound, there exists $X\in\ell^2$, $X\neq0$, such that
\[
(\hm T_\theta(\omega_m)-\omega_{\min}\hm I)X=0.
\]
In the time domain, this condition implies $(\omega_m(t)-\omega_{\min})x(t)=0$ almost everywhere on $[t-T(t),\,t]$. Since $x$ is not identically zero and $\omega_m$ is continuous, it follows that $\omega_m(t)=\omega_{\min}$ on this interval. Therefore, the maximal level set is obtained by considering boundary trajectories satisfying $\omega_m(t)=\omega_{\min}$ or $\omega_m(t)=\omega_{\max}$.

Because the Lyapunov function defines convex ellipsoidal level sets centered at the equilibrium, the active constraint is necessarily the boundary closest to the equilibrium speed $\Omega^{\mathrm{ref}}$ (see Fig.~\ref{fig:ligneniveau}). 

\begin{figure}
    \centering
    \includegraphics[width=0.9\linewidth]{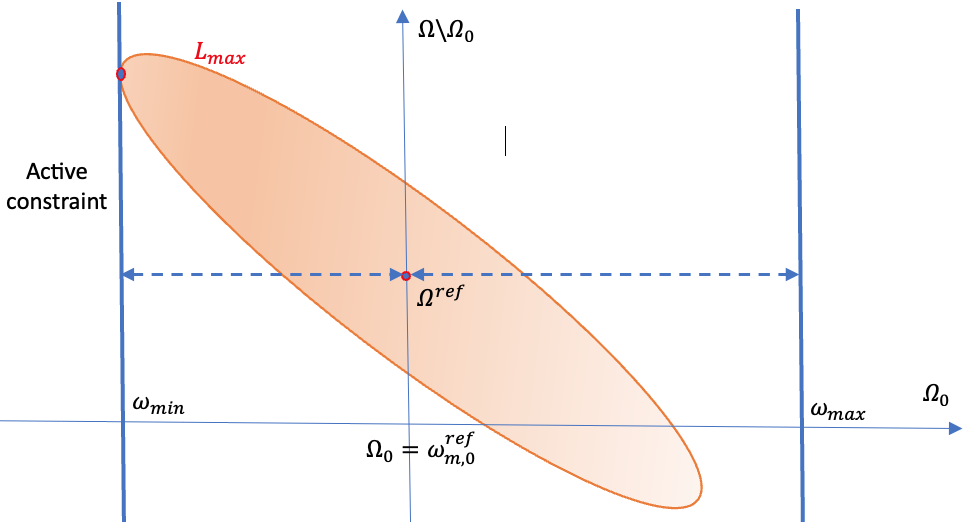}
    \caption{Projection of the maximum level set on the subspace $\Omega\setminus\Omega_0\times \Omega_0 $}\label{fig:ligneniveau}
\end{figure} 

Therefore, the admissible level is determined by
\[
\delta(\omega_{m,0}^{\mathrm{ref}})
=
\min\!\bigl(
|\omega_{\min}-\omega_{m,0}^{\mathrm{ref}}|,
|\omega_{\max}-\omega_{m,0}^{\mathrm{ref}}|
\bigr).
\]
Without loss of generality, assume that the minimum is attained at $\omega_{\min}$.

Let $\Omega=\hm F_\theta(\omega_m)$. Since $\omega_m$ is constant on the constraint boundary, we set $\Omega_{k}=0$ for all $k\neq0$ in the augmented state ${\widetilde E}$. The computation of the maximal level set then reduces to the quadratic optimization problem
\[
\max_{X,Z}\; \hm V({\widetilde E})
\]
subject to the constraint $ C_0(\Omega_{0})=\Omega_{0}-\omega_{\min}\ge0$.
Using the KKT optimality conditions, there exists a Lagrange multiplier $\alpha$ such that
\[
{\widetilde E}=\frac{1}{2}\alpha \hm S \nabla_{[X,Z]} C_0(\Omega_{0}),
\]
where $\hm S=\hm P^{-1}$ and $\nabla_{[X,Z]} C_0(\Omega_{m,0})$ denotes the gradient of the constraint with respect to $(X,Z)$. Considering the zero-order phasor associated with $\Omega$ (i.e., the central row of the fourth Toeplitz block), we obtain
\[
\Omega_{0}-\omega_{m,0}^{\mathrm{ref}}
=
\frac{1}{2}\alpha {\bigl(\hm S \nabla_{[X,Z]} C_0(\Omega_{0})\bigr)}_{4,0}.
\]

Since $\nabla_{[X,Z]} C_0(\Omega_{0})$ is zero everywhere except for the component associated with $\Omega_{0}$, which is equal to $1$, the Toeplitz structure of $\hm S$ yields
\[
{(\hm S \nabla_{[X,Z]} C_0(\Omega_{0}))}_{4}
=
{S}_{44}
=
\hm F_\theta(s_{44}).
\]
Therefore,
\[
\Omega_{0}-\omega_{m,0}^{\mathrm{ref}}
=
\frac{1}{2}\alpha\,{S}_{44,0}.
\]
Enforcing the boundary condition $\Omega_{0}=\omega_{\min}$ gives
\[
\alpha
=
2\,\frac{\omega_{\min}-\omega_{m,0}^{\mathrm{ref}}}{{S}_{44,0}}.
\]
Substituting the corresponding expression of ${\widetilde E}$ into the Lyapunov function yields~\eqref{eqLmax}, which concludes the proof.
\hfill $\square$
\end{pf}

\bibliographystyle{unsrt}
\bibliography{references} 
 
\end{document}